\documentclass[a4paper,UKenglish,cleveref, autoref, thm-restate]{lipics-v2021}

\hideLIPIcs  

\usepackage[ruled,vlined,linesnumbered]{algorithm2e}
\usepackage{enumitem}

\input{styles/problem_box.sty}
\input{styles/variables.sty}
\input{styles/utilities.sty}
\newtheorem{alg}{Algorithm}

\graphicspath{{./images/}}

\title{Improved Parameterized Complexity of Happy Set Problems} 

\newcommand{\universityOfUtah}{School of Computing, University of Utah, USA}
\author{Yosuke Mizutani}{\universityOfUtah}{yos@cs.utah.edu}{https://orcid.org/0000-0002-9847-4890}{}
\author{Blair D. Sullivan}{\universityOfUtah}{sullivan@cs.utah.edu}{https://orcid.org/0000-0001-7720-6208}{}

\authorrunning{Y. Mizutani and B.\,D. Sullivan} 

\Copyright{Yosuke Mizutani and Blair D. Sullivan} 



\category{} 

\relatedversion{} 

\supplement{}

\funding{This work was supported in part by the Gordon \& Betty Moore
Foundation under award GBMF4560 to Blair D. Sullivan.}

\acknowledgements{We thank
Arnab Banerjee,
Oliver Flatt and
Thanh Son Nguyen
for their contributions to a course project that led to this research.
}

\nolinenumbers 

\EventEditors{John Q. Open and Joan R. Access}
\EventNoEds{2}
\EventLongTitle{42nd Conference on Very Important Topics (CVIT 2016)}
\EventAcronym{CVIT}
\EventYear{2016}
\EventDate{December 24--27, 2016}
\EventLocation{Little Whinging, United Kingdom}
\EventLogo{}
\SeriesVolume{42}
\ArticleNo{23}

\begin{document}

\maketitle
\begin{abstract}
    We present fixed-parameter tractable (FPT) algorithms for two problems,
     \PrbMaxHSLong (\PrbMaxHS) and \PrbMaxEHSLong (\PrbMaxEHS)%
     ---also known as \PrbDKSLong.
    Given a graph $G$ and an integer $k$, \PrbMaxHS asks for a set $S$ of
     $k$ vertices such that the number of \textit{happy vertices} with respect to $S$ is maximized,
     where a vertex $v$ is happy if $v$ and all its neighbors are in $S$.
    We show that \PrbMaxHS can be solved in time
     $\mathcal{O}\left(2^\mw \cdot \mw \cdot k^2 \cdot \abs{V(G)}\right)$
     and $\mathcal{O}\left(8^\cw \cdot k^2 \cdot \abs{V(G)}\right)$,
     where $\mw$ and $\cw$ denote the \textit{modular-width} and the \textit{clique-width}
     of $G$, respectively.
    This resolves the open questions posed in \cite{asahiro2021parameterized}.

    The \PrbMaxEHS problem is an edge-variant of \PrbMaxHS,
     where we maximize the number of \textit{happy edges}, the edges whose endpoints are in $S$.
    In this paper we show that \PrbMaxEHS can be solved in time
    $f(\nd)\cdot\abs{V(G)}^{\bigo{1}}$ and $\bigo{2^{\cd}\cdot k^2 \cdot \abs{V(G)}}$,
     where $\nd$ and $\cd$ denote the \textit{neighborhood diversity} and
     the \textit{cluster deletion number} of $G$, respectively,
     and $f$ is some computable function.
    This result implies that \PrbMaxEHS is also fixed-parameter tractable by \textit{twin cover number}.
\end{abstract}

\newpage
\section{Introduction}\label{sec:introduction}

In the study of large-scale networks, \textit{communities}---%
cohesive subgraphs in a network---play an important role in understanding complex systems
and appear in sociology, biology and computer science, etc. \cite{fortunato_community_2010,li_community_2011}.
For example, the concept of homophily in sociology explains the tendency for individuals
 to associate themselves with similar people \cite{mcpherson2001birds}.
Homophily is a fundamental law governing
 the structure of social networks,
 and finding groups of people sharing similar interests
 has many real-world applications \cite{easley2012networks}.

People have attempted to frame this idea as a graph optimization problem,
where a vertex represents a person and an edge corresponds to some relation in the social network.
The notion of \textit{happy} vertices was first introduced by Zhang and Li
 in terms of graph coloring \cite{zhang2015algorithmic},
 where each color represents an attribute of a person (possibly fixed).
A vertex is \textit{happy} if all of its neighbors share its color.
The goal is to maximize the number of happy vertices by changing the color of
 unfixed vertices, thereby achieving the greatest social benefit.

Later, Asahiro et al. introduced \PrbMaxHSLong (\PrbMaxHS)
which defines that a vertex $v$ is \textit{happy} with respect to a \textit{happy set} $S$
if $v$ and all of its neighbors are in $S$ \cite{asahiro2021parameterized}.
The \PrbMaxHS problem asks for a vertex set $S$ of size $k$ that maximizes
 the number of happy vertices.
They also define its edge-variant, \PrbMaxEHSLong (\PrbMaxEHS)
 which maximizes the number of \textit{happy edges},
 an edge with both endpoints in the \textit{happy set}.
It is clear to see that \PrbMaxEHS is equivalent to choosing a vertex set $S$ such that
the number of edges in the induced subgraph on $S$ is maximized.
This problem is known as \PrbDKSLong (\PrbDKS) in other literature.
Both \PrbMaxHS and \PrbMaxEHS are NP-hard \cite{asahiro2021parameterized, feige_densest_1997},
and we study their parameterized complexity throughout this paper.

\subsection{Parameterized Complexity and Related Work}

Graph problems are often studied with a variety of structural parameters
 in addition to natural parameters (size $k$ of the happy set in our case). 
Specifically, we investigate the parameterized complexity with respect to
 \textit{modular-width} (\mw),
 \textit{clique-width} (\cw),
 \textit{neighborhood diversity} (\nd),
 \textit{cluster deletion number} (\cd),
 \textit{twin cover number} (\tc),
 \textit{treewidth} (\tw),
 and \textit{vertex cover number} (\vc),
 all of which we define in Section 2.3.
Figure~\ref{fig:params} illustrates the hierarchy of these parameters by inclusion;
 hardness results are implied along the arrows,
 and FPT\footnote{%
 An FPT (fixed-parameter tractable) algorithm solves the problem in time $f(k) \cdot n^{\bigo{1}}$ for
 some computable function $f$.
 } algorithms are implied in the reverse direction.

Asahiro et al. showed that \PrbMaxHS is W[1]-hard with respect to $k$ by a parameterized reduction
 from the \probname{$q$-Clique} problem \cite{asahiro2021parameterized}.
They also presented FPT algorithms for \PrbMaxHS on parameters:
 clique-width plus $k$, neighborhood diversity, cluster deletion number (which implies FPT by twin cover number),
 and treewidth.

\PrbMaxEHS (\PrbDKS) has been extensively studied in different names
 (e.g. the \probname{$k$-Cluster} problem \cite{corneil_clustering_1984},
  the \probname{Heaviest Unweighted Subgraph} problem \cite{kortsarz_choosing_1993},
  and \probname{$k$-Cardinality Subgraph} problem \cite{bruglieri_annotated_2006}).
As for parameterized complexity, Cai showed the W[1]-hardness parameterized by $k$ \cite{cai2007parameterized}.
Bourgeois employed Moser's technique in \cite{moser2005exact} to show that
\PrbMaxEHS can be solved in time $2^{\tw}\cdot n^{\bigo{1}}$ \cite{bourgeois2013exact}.
Broersma et al. proved that \PrbMaxEHS can be solved in time $k^{\bigo{\cw}}\cdot n$,
but it cannot be solved in time $2^{o(\cw \log k)} \cdot n^{\bigo{1}}$
unless the Exponential Time Hypothesis (ETH) fails \cite{broersma2013tight}.
To the best of our knowledge, the parameterized hardness by modular-width,
 neighborhood diversity, cluster deletion number and twin cover number remained open prior to our work.
Figure~\ref{tab:result} summarizes the known and established hardness results for \PrbMaxHS and \PrbMaxEHS.

\subsection{Our Contributions}

In this paper, we present four novel parameterized algorithms for \PrbMaxHS and \PrbMaxEHS.
First, we give a dynamic-programming algorithm that solves \PrbMaxHS in time\linebreak
$\bigo{2^\mw \cdot \mw \cdot k^2 \cdot |V|}$,
answering the question posed by the authors of \cite{asahiro2021parameterized}.
Second, we show that \PrbMaxHS is FPT by clique-width, giving
an $\bigo{8^\cw \cdot k^2 \cdot |V|}$ algorithm,
which removes the exponential term of $k$ from the best known result,
$\bigo{6^{\cw} \cdot k^{2(\cw + 1)} \cdot |V|}$ \cite{asahiro2021parameterized}.
Turning to \PrbMaxEHS, we prove it is FPT by neighborhood diversity,
using an integer quadratic programming formulation.
Lastly, we provide an FPT algorithm for \PrbMaxEHS parameterized by cluster deletion number
with running time $\bigo{2^{\cd} \cdot k^2 \cdot |V|}$,
which also implies the problem is FPT by twin cover number.
These new results complete the previously-open parameterized complexities in \cref{tab:result}.

\begin{figure}[t]
    \begin{minipage}{0.40\textwidth}
     \centering
     \footnotesize
     \begin{tikzpicture}
        \tikzstyle{Rect}=[rectangle, draw, align=center, thick];

        \node[Rect] (cw) at (0.5,4   ) {clique-width};
        \node[Rect] (mw) at (-2 ,2.8 ) {modular-width};
        \node[Rect] (cd) at (0  ,2.65) {cluster\\deletion};
        \node[Rect] (nd) at (-2 ,1.5 ) {neighborhood\\diversity};
        \node[Rect] (tc) at (0  ,1.5 ) {twin cover};
        \node[Rect] (vc) at (0  ,0   ) {vertex cover};
        \node[Rect] (tw) at (1.4,0.9 ) {treewidth};

        \node (weak)   at (-2.3, 4) {(weak)};
        \node (strong) at (-2.3, 0) {(strong)};
     
        \draw[-{Latex[length=2mm, width=1.5mm]}] (vc)--(nd);
        \draw[-{Latex[length=2mm, width=1.5mm]}] (vc)--(tc);
        \draw[-{Latex[length=2mm, width=1.5mm]}] (vc)--(tw);
        \draw[-{Latex[length=2mm, width=1.5mm]}] (nd)--(mw);
        \draw[-{Latex[length=2mm, width=1.5mm]}] (tc)--(mw);
        \draw[-{Latex[length=2mm, width=1.5mm]}] (tc)--(cd);
        \draw[-{Latex[length=2mm, width=1.5mm]}] (cd)--(cw);
        \draw[-{Latex[length=2mm, width=1.5mm]}] (tw)--(0.9, 3.75);
        \draw[-{Latex[length=2mm, width=1.5mm]}] (mw)--(-0.1, 3.75);

     \end{tikzpicture}

     \caption{%
     Hierarchy of relevant structural graph parameters.
     Arrows indicate generalizations.
     }
     \label{fig:params}
    \end{minipage}%
    \hspace{0.02\textwidth}
    \begin{minipage}{0.56\textwidth}
     \centering
     \footnotesize
     \setlength{\tabcolsep}{0.6em}
     \begin{tabular}{|l|c|c|}
         \hline
         \rule{0pt}{10pt}
         \textbf{Parameter} & \PrbMaxHS & \PrbMaxEHS \\
         \hline\hline\rule{0pt}{9pt}
         Size $k$ of happy set    & W[1]-hard\cite{asahiro2021parameterized} & W[1]-hard\cite{cai2007parameterized} \\
         \hline\rule{0pt}{9pt}
         Clique-width + $k$           & FPT\cite{asahiro2021parameterized} & FPT\cite{broersma2013tight} \\
         \hline\hline\rule{0pt}{9pt}
         Clique-width                 & \textcolor{red}{FPT} & W[1]-hard\cite{broersma2013tight} \\
         \hline\rule{0pt}{9pt}
         Modular-width                & \textcolor{red}{FPT} & \textcolor{blue}{Open} \\
         \hline\rule{0pt}{9pt}
         \scriptsize{Neighborhood diversity}       & FPT\cite{asahiro2021parameterized} & \textcolor{red}{FPT} \\
         \hline\rule{0pt}{9pt}
         \scriptsize{Cluster deletion number}      & FPT\cite{asahiro2021parameterized} & \textcolor{red}{FPT} \\
         \hline\rule{0pt}{9pt}
         Twin cover number            & FPT\cite{asahiro2021parameterized} & \textcolor{red}{FPT} \\
         \hline\rule{0pt}{9pt}
         Treewidth                    & FPT\cite{asahiro2021parameterized} & FPT\cite{bourgeois2013exact}\\
         \hline\rule{0pt}{9pt}
         Vertex cover number          & FPT\cite{asahiro2021parameterized} & FPT\cite{bourgeois2013exact}\\
         \hline
     \end{tabular}
     \caption{Known and established hardness results under select parameters
     for \PrbMaxHS and \PrbMaxEHS (as known as \PrbDKSLong).
     New results from this paper in red.}
     \label{tab:result}
    \end{minipage}
\end{figure}


\section{Preliminaries}\label{sec:preliminaries}

We use standard graph theory notation, following \cite{diestel2018graph}.
Given a graph $G=(V,E)$,
we write $n(G)=\abs{V}$ for the number of vertices and
$m(G)=\abs{E}$ for the number of edges.
We use $N(v)$ and $N[v]$ to denote the open and closed
neighborhoods of a vertex $v$, respectively, and
for a vertex set $X \subseteq V$, $N[X]$ denotes the union of $N[x]$ for all $x \in X$.
We write $\deg_G(v)=\deg(v)$ for the degree of a vertex $v$.
We denote the induced subgraph of $G$ on a set $X \subseteq V$ by $G[X]$.
%
%
We say vertices $u$ and $v$ are \textit{twins} if they have the same neighbors,
 i.e. $N(u)\setminus \{v\} = N(v) \setminus \{u\}$.
Further, they are called \textit{true twins} if $uv \in E$.

\subsection{Problem Definitions}

Asahiro et al. first introduced the \PrbMaxHSLong problem in \cite{asahiro2021parameterized}.

\begin{problem}{\PrbMaxHSLong (\PrbMaxHS)}
\Input & A graph $G = (V, E)$ and a positive integer $k$. \\
\Prob & Find a subset $S \subseteq V$ of $k$ vertices that maximizes
 the number of \textit{happy} vertices $v$ with $N[v] \subseteq S$.
\end{problem}

\begin{figure*}[t]
    \centering

    \pgfdeclarelayer{bg}
    \pgfsetlayers{bg, main}
    \usetikzlibrary{calc}

    \tikzset{
        old inner xsep/.estore in=\oldinnerxsep,
        old inner ysep/.estore in=\oldinnerysep,
        double circle/.style 2 args={
            circle,
            old inner xsep=\pgfkeysvalueof{/pgf/inner xsep},
            old inner ysep=\pgfkeysvalueof{/pgf/inner ysep},
            /pgf/inner xsep=\oldinnerxsep+#1,
            /pgf/inner ysep=\oldinnerysep+#1,
            alias=sourcenode,
            append after command={
            let     \p1 = (sourcenode.center),
                    \p2 = (sourcenode.east),
                    \n1 = {\x2-\x1-#1-0.5*\pgflinewidth}
            in
                node [inner sep=0pt, draw, circle, minimum width=2*\n1,at=(\p1),#2] {}
            }
        },
        double circle/.default={-3pt}{black}
    }
    \tikzstyle{plain} = [circle, fill=white, text=black, draw, thick, scale=1, minimum size=0.5cm, inner sep=1.5pt]
    \tikzstyle{happy} = [double circle={-3.5pt}{line width=1.2pt}, fill=black!25, text=black, draw, thick, scale=1, minimum size=0.5cm, inner sep=1.5pt]
    \tikzstyle{unhappy} = [circle, fill=black!25, text=black, draw, thick, scale=1, minimum size=0.5cm, inner sep=1.5pt]

    \begin{minipage}[m]{.48\textwidth}
        \vspace{0pt}
        \centering
        \begin{tikzpicture}
            \node[plain] (a) at (0, 1) {$a$};
            \node[plain] (b) at (0, 0) {$b$};
            \node[unhappy] (c) at (1.3, 1) {$c$};
            \node[unhappy] (d) at (1.3, 0) {$d$};
            \node[happy] (e) at (2.6, 0.5) {$e$};
            \node[unhappy] (f) at (3.9, 1) {$f$};
            \node[unhappy] (g) at (3.9, 0) {$g$};
            \node[plain] (h) at (4.8, 0.5) {$h$};

            \draw (a) -- (b);
            \draw (a) -- (c);
            \draw (a) -- (d);
            \draw (b) -- (c);
            \draw (b) -- (d);
            \draw (c) -- (d);
            \draw (c) -- (e);
            \draw (d) -- (e);
            \draw (e) -- (f);
            \draw (e) -- (g);
            \draw (f) -- (h);
            \draw (g) -- (h);

            \draw[rounded corners, dashed] (0.9, -0.5) rectangle ++ (3.4, 2);
            \node at (3.2, 1.2) {$S$};
        \end{tikzpicture}
    \end{minipage}
    \begin{minipage}[m]{.48\textwidth}
        \vspace{0pt}
        \centering
        \begin{tikzpicture}
            \node[happy] (a) at (0, 1) {$a$};
            \node[happy] (b) at (0, 0) {$b$};
            \node[happy] (c) at (1.3, 1) {$c$};
            \node[happy] (d) at (1.3, 0) {$d$};
            \node[unhappy] (e) at (2.6, 0.5) {$e$};
            \node[plain] (f) at (3.9, 1) {$f$};
            \node[plain] (g) at (3.9, 0) {$g$};
            \node[plain] (h) at (4.8, 0.5) {$h$};

            \draw (a) -- (b);
            \draw (a) -- (c);
            \draw (a) -- (d);
            \draw (b) -- (c);
            \draw (b) -- (d);
            \draw (c) -- (d);
            \draw (c) -- (e);
            \draw (d) -- (e);
            \draw (e) -- (f);
            \draw (e) -- (g);
            \draw (f) -- (h);
            \draw (g) -- (h);

            \draw[rounded corners, dashed] (-0.5, -0.5) rectangle ++ (3.5, 2);
            \node at (2.2, 1.2) {$S$};
        \end{tikzpicture}
    \end{minipage}
    \caption{%
    Given a graph above, if $k=5$,
    choosing $S=\{c, d, e, f, g\}$ makes only one vertex ($e$) happy (left).
    On the other hand, $S=\{a, b, c, d, e\}$ is an optimal solution,
    making four vertices ($a, b, c, d$) happy (right).
    }
    \label{fig:hs-ex}
\end{figure*}
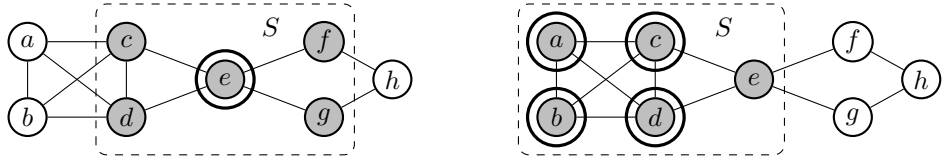

Figure~\ref{fig:hs-ex} illustrates an example instance with $k=5$.
Let us call a vertex that is not happy an \textit{unhappy} vertex,
 and observe that the set of unhappy vertices is given by $N[V \setminus S]$,
 providing an alternative characterization of happy vertices.

\begin{proposition}\label{prop:unhappy}
    Given a graph $G=(V,E)$ and a happy set $S \subseteq V$,
    the set of happy vertices is given by $V \setminus N[V \setminus S]$.
\end{proposition}

In addition, Asahiro et al. define an edge variant \cite{asahiro2021parameterized}:

\begin{problem}{\PrbMaxEHSLong (\PrbMaxEHS)}
\Input & A graph $G=(V,E)$ and a positive integer $k$.\\
\Prob & Find a set $S \subseteq V$ of $k$ vertices that maximizes the number of happy edges.
        An edge $uv \in E$ is \textit{happy} if and only if $\{u,v\} \subseteq S$.
\end{problem}

It is known that \PrbMaxEHS is identical to the \PrbDKSLong problem (\PrbDKS),
as the number of happy edges is equal to $m(G[S])$.
Some literature (e.g. \cite{fomin_algorithmic_2010})
also phrases this problem as the dual of the \PrbSKSLong problem.

\subsection{Structural Parameters}

We now define the structural graph parameters considered in this paper.

\br
\noindent\textbf{Treewidth.}
The most-studied structural parameter is treewidth,
 introduced by Robertson \& Seymour in \cite{robertson_graph_1986}.
Treewidth measures how a graph resembles a tree and
 admits FPT algorithms for a number of NP-hard problems, such as
 \textsc{Weighted Independent Set}, \textsc{Dominating Set}, and \textsc{Steiner Tree} \cite{cygan2015parameterized}.
Treewidth is defined by the following notion of tree decomposition.
\begin{definition}[treewidth \cite{cygan2015parameterized}]
    A tree decomposition of a graph $G$ is a pair \\$(T, \{X_t\}_{t \in V(T)})$,
    where $T$ is a tree and $X_t \subseteq V(G)$ is an assigned vertex set for every node $t$,
    such that the following three conditions hold:
    \vspace*{-0.5em}
    \begin{itemize}
        \item $\bigcup_{t \in V(T)} X_t = V(G)$.
        \item For every $uv \in E(G)$, there exists a node $t$ such that $u,v \in X_t$.
        \item For every $u \in V(G)$, the set $T_u = \{t \in V(T) : u \in X_t \}$
        induces a connected subtree of $T$.
    \end{itemize}

    The width of tree decomposition is defined to be $\max_{t\in V(T)}\abs{X_t}-1$,
    and the \textbf{treewidth} of a graph $G$, denoted by $\tw$,
     is the minimum possible width of a tree decomposition of $G$.
\end{definition}

\noindent\textbf{Clique-width.}
Clique-width is a generalization of treewidth and can capture 
 dense, but structured graphs.
Intuitively, a graph with bounded clique-width $k$ can be built from
 single vertices by joining structured parts,
 where vertices are associated by at most $k$ labels such that
 those with the same label are indistinguishable in later steps.

\begin{definition}[clique-width \cite{courcelle2000upper}]
  For a positive integer $w$, a $w$-labeled graph is a graph whose vertices are
   labeled by integers in $\{1,\ldots,w\}$.
  The \textbf{clique-width} of a graph $G$, denoted by $\cw$, is the minimum $w$ such that
  $G$ can be constructed by repeated application of the following operations:

  \vspace*{-0.5em}
  \begin{itemize}
    \item (O1) Introduce $i(v)$: add a new vertex $v$ with label $i \in \{1,\ldots,w\}$.
    \item (O2) Union $G_1 \oplus G_2$: take a disjoint union of $w$-labeled graphs $G_1$ and $G_2$.
    \item (O3) Join $\eta(i,j)$: take two labels $i$ and $j$, and then add an edge between
                every pair of vertices labeled by $i$ and by $j$.
    \item (O4) Relabel $\rho(i,j)$: relabel the vertices of label $i$ to label $j \in \{1,\ldots,w\}$.
  \end{itemize}
\end{definition}

This construction naturally defines a rooted binary tree, called a $\cw$-expression tree $G$,
 where $G$ is the root and each node corresponds to
 one of the above operations.

\br
\noindent\textbf{Neighborhood Diversity.}
Neighborhood diversity is a parameter introduced by Lampis \cite{lampis_algorithmic_2012}, which measures the number of twin classes.

\begin{definition}[neighborhood diversity \cite{lampis_algorithmic_2012}]
    The \textbf{neighborhood diversity} of a graph $G=(V,E)$, denoted by $\nd$,
     is the minimum number $w$ such that $V$ can be partitioned into $w$ sets of twin vertices.
\end{definition}

By definition, each set of twins, called a \textit{module}, is either a clique or an independent set.

\br
\noindent\textbf{Cluster Deletion Number.}
Cluster (vertex) deletion number is the distance to a cluster graph, which consists of disjoint cliques.

\begin{definition}[cluster deletion number]
    A vertex set $X$ is called a cluster deletion set if $G[V \setminus X]$ is a cluster graph.
    The \textbf{cluster deletion number} of $G$, denoted by $\cd$, is the size of the minimum cluster deletion set in $G$.
\end{definition}

\br
\noindent\textbf{Twin Cover Number.}
The notion of twin cover is introduced by Ganian \cite{ganian_improving_2015} and offers a generalization
 of vertex cover number.

\begin{definition}[twin cover number \cite{ganian_improving_2015}]
    A vertex set $X \subseteq V$ is a twin cover of $G=(V,E)$
     if for every edge $uv \in E$ either
     (1) $u \in X$ or $v \in X$, or
     (2) $u$ and $v$ are true twins.
    The \textbf{twin cover number}, denoted by $\tc$, is the size of the minimum twin cover of $G$.
\end{definition}

\br
\noindent\textbf{Modular-width.}
Modular-width is a parameter introduced by Gajarsk{\'y} et al. \cite{modular-width-gajarsky}
 to generalize simpler notions on dense graphs
 while avoiding the negative results brought by moving to the full generality of clique-width
 (e.g. many problems FPT for treewidth becomes W[1]-hard for clique-width \cite{fomin_clique-width_2009,fomin_algorithmic_2010,fomin_intractability_2010}).
Modular-width is defined using the standard concept of modular decomposition.

\begin{definition}[modular-width \cite{modular-width-gajarsky}]
    Any graph can be produced via a sequence of the following operations:

    \vspace*{-0.5em}
    \begin{itemize}
      \item (O1) Introduce: Create an isolated vertex.
      \item (O2) Union $G_1 \oplus G_2$: Create the disjoint union of two graphs $G_1$ and $G_2$.
      \item (O3) Join: Given two graphs $G_1$ and $G_2$, create the complete join $G_3$ of  $G_1$ and $G_2$.
                That is, a graph $G_3$ with vertices $V(G_1) \cup V(G_2)$ and edges $E(G_1) \cup E(G_2) \cup \{(v, w) : v \in G_1, w \in G_2\}$.
      \item (O4) Substitute: Given a graph $G$ with vertices $v_1,\ldots,v_n$ and given graphs $G_1,\ldots,G_n$, create the \textit{substitution}
                of $G_1,\ldots,G_n$ in $G$. The substitution is a graph $\mathcal{G}$ with vertex set $\bigcup_{1\leq i \leq n} V(G_i)$
                and edge set $\bigcup_{1\leq i \leq n}{E(G_i)} \cup \{(v, w) : v \in G_i, w \in G_j, (v_i, v_j) \in E(G)\}$.
                Each graph $G_i$ is substituted for a vertex $v_i$,
                 and all edges between graphs corresponding to adjacent vertices in $G$ are added.
    \end{itemize}
    
    These operations, taken together in order to construct a graph, 
    form a \textit{parse-tree} of the graph.
    The width of a graph is the maximum size of the vertex set of $G$
     used in operation (O4) to construct the graph.
     The \textbf{modular-width}, denoted by $\mw$, is the minimum width such that
     $G$ can be obtained from some sequence of operations (O1)-(O4).
\end{definition}

Finding a parse-tree of a given graph, called a \textit{modular decomposition}, can be done in linear-time \cite{tedder_simple_2007}.
See Figure~\ref{fig:mw-ex} for an illustration of modular decomposition. 
Gajarský et al. also give FPT algorithms parameterized by modular-width for 
 \textsc{Partition into paths, Hamiltonian path, Hamiltonian cycle} and \textsc{Coloring},
 using bottom-up dynamic programming along the parse-tree.

\begin{figure*}[t]
    \centering

    \pgfdeclarelayer{bg}
    \pgfsetlayers{bg, main}
    \usetikzlibrary{calc}

    \tikzset{
        old inner xsep/.estore in=\oldinnerxsep,
        old inner ysep/.estore in=\oldinnerysep,
        double circle/.style 2 args={
            circle,
            old inner xsep=\pgfkeysvalueof{/pgf/inner xsep},
            old inner ysep=\pgfkeysvalueof{/pgf/inner ysep},
            /pgf/inner xsep=\oldinnerxsep+#1,
            /pgf/inner ysep=\oldinnerysep+#1,
            alias=sourcenode,
            append after command={
            let     \p1 = (sourcenode.center),
                    \p2 = (sourcenode.east),
                    \n1 = {\x2-\x1-#1-0.5*\pgflinewidth}
            in
                node [inner sep=0pt, draw, circle, minimum width=2*\n1,at=(\p1),#2] {}
            }
        },
        double circle/.default={-3pt}{black}
    }
    \tikzstyle{plain} = [circle, fill=white, text=black, draw, thick, scale=1, minimum size=0.5cm, inner sep=1.5pt]
    \tikzstyle{small} = [circle, fill=white, text=black, draw, thick, scale=1, minimum size=0.2cm, inner sep=1.5pt]
    \tikzstyle{large} = [circle, fill=black!25, text=black, draw, thick, scale=1, minimum size=1.0cm, inner sep=1.5pt]

    \begin{minipage}[m]{.32\textwidth}
        \vspace{0pt}
        \centering
        \begin{tikzpicture}
            \node[plain] (a) at (0, 1) {$a$};
            \node[plain] (b) at (1, 1) {$b$};
            \node[plain] (c) at (1.8, 2) {$c$};
            \node[plain] (d) at (1.8, 0) {$d$};
            \node[plain] (e) at (2.8, 1) {$e$};
            \node[plain] (f) at (3.6, 2) {$f$};
            \node[plain] (g) at (3.6, 0) {$g$};

            \draw (a) -- (b);
            \draw (a) -- (c);
            \draw (a) -- (d);
            \draw (b) -- (e);
            \draw (c) -- (d);
            \draw (c) -- (e);
            \draw (d) -- (e);
            \draw (e) -- (f);
            \draw (e) -- (g);

            \draw[rounded corners, dashed, blue] (0.6, -0.5) rectangle ++ (1.8, 3);
            \draw[rounded corners, dashed, black!60!green] (1.4, -0.4) rectangle ++ (0.8, 2.8);
            \draw[rounded corners, dashed, black!20!orange] (3.2, -0.5) rectangle ++ (0.8, 3);
        \end{tikzpicture}
    \end{minipage}
    \begin{minipage}[m]{.66\textwidth}
        \vspace{0pt}
        \centering
        \begin{tikzpicture}
            \node[plain] (a0) at (0, 1.5) {$a$};
            \node[plain] (b0) at (1, 1.5) {$b$};
            \node[plain] (c0) at (2, 1.5) {$c$};
            \node[plain] (d0) at (3, 1.5) {$d$};
            \node[plain] (e0) at (4, 1.5) {$e$};
            \node[plain] (f0) at (5, 1.5) {$f$};
            \node[plain] (g0) at (6, 1.5) {$g$};

            \node[small] (c1) at (2.5, 2.8) {};
            \node[small] (d1) at (2.5, 2.4) {};

            \node[small] (b2) at (1.4, 4.2) {};
            \node[large] (cy) at (2.4, 4.2) {};
            \node[small] (c2) at (2.4, 4.4) {};
            \node[small] (d2) at (2.4, 4.0) {};
            \node[small] (f2) at (5.5, 4.4) {};
            \node[small] (g2) at (5.5, 4.0) {};

            \node[large] (ax) at (0.5,  6.0) {};
            \node[small] (a3) at (0.5,  6.0) {};
            \node[large] (bx) at (2.05, 6.0) {};
            \node[small] (b3) at (1.8,  6.0) {};
            \node[small] (c3) at (2.2,  6.2) {};
            \node[small] (d3) at (2.2,  5.8) {};
            \node[large] (ex) at (4.0,  6.0) {};
            \node[small] (e3) at (4.0,  6.0) {};
            \node[large] (fx) at (5.5,  6.0) {};
            \node[small] (f3) at (5.5,  6.2) {};
            \node[small] (g3) at (5.5,  5.8) {};

            \draw (c1) -- (d1);
            \draw (c2) -- (d2);
            \draw (c3) -- (d3);
            \draw[very thick] (ax) -- (bx);
            \draw[very thick] (bx) -- (ex);
            \draw[very thick] (ex) -- (fx);
            \draw (a0) -- (0, 5.3);
            \draw (e0) -- (4, 5.3);
            \draw (f0) -- (5, 3.5);
            \draw (g0) -- (6, 3.5);
            \draw (5.5, 4.9) -- (5.5, 5.3);
            \draw (2.05, 4.9) -- (2.05, 5.3);
            \draw (2.5, 3.1) -- (2.5, 3.5);
            \draw (b0) -- (1, 3.5);
            \draw (c0) -- (2, 2.1);
            \draw (d0) -- (3, 2.1);

            \draw[rounded corners, dashed] (-0.2, 5.3) rectangle ++ (6.4, 1.4);
            \draw[rounded corners, dashed, blue] (0.8, 3.5) rectangle ++ (2.4, 1.4);
            \draw[rounded corners, dashed, black!20!orange] (4.8, 3.5) rectangle ++ (1.4, 1.4);
            \draw[rounded corners, dashed, black!60!green] (1.8, 2.1) rectangle ++ (1.4, 1.0);
            \node at (6.7, 6.4) {(O4)};
            \node at (6.7, 4.6) {(O2)};
            \node at (3.6, 4.6) {(O2)};
            \node at (3.6, 2.8) {(O3)};
            \node at (6.7, 1.8) {(O1)};
        \end{tikzpicture}
    \end{minipage}
    \caption{%
    An example graph $G$ (left) with modular-width $4$.
    Modular decomposition of the same graph is shown at right.
    The parse-tree has $G$ as the root, and its nodes correspond to operations (O1)-(O4).
    Notice that each node also represents a \textit{module}%
     ---module members have the same neighbors outside the module.
    }
    \label{fig:mw-ex}
\end{figure*}
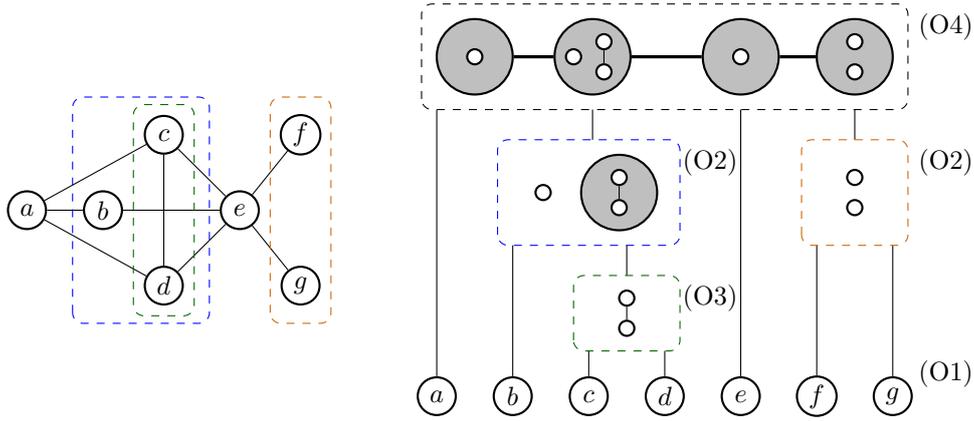

\br
\noindent\textbf{Vertex Cover Number.}
The vertex cover number is the solution size of the classic \textsc{Vertex Cover} problem.

\begin{definition}
    A vertex set $X \subseteq V$ is a vertex cover of $G=(V,E)$ if for every edge $uv \in E$ either $u \in X$ or $v \in X$.
    The \textbf{vertex cover number} of $G$, denoted by $\vc$, is the size of the minimum vertex cover of $G$.
\end{definition}

\subsubsection{Properties of Structural Parameters}

Finally, we note the relationship among the parameters defined above,
 which establishes the hierarchy shown in Figure~\ref{fig:params}.

\begin{proposition}[\cite{asahiro2021parameterized, modular-width-gajarsky}]
    Let $\cw, \tw, \cd, \nd, \tc, \vc, \mw$ be
     the clique-width,
     tree-width,
     cluster deletion number,
     neighborhood diversity,
     twin cover number,
     vertex cover number,
     and modular-width
     of a graph $G$, respectively.
    Then the following inequalities hold\footnote{%
    $\cw \leq 2$ when $\mw = 0$; otherwise, $\cw \leq \mw + 1$.
}:
     (i) $\cw \leq 2^{\tw + 1} + 1$;
     (ii) $\tw \leq \vc$;
     (iii) $\nd \leq 2^{\vc} + \vc$;
     (iv) $\cw \leq 2^{\cd + 3} - 1$;
     (v) $\cd \leq \tc \leq \vc$;
     (vi) $\mw \leq \nd$;
     (vii) $\mw \leq 2^{\tc} + \tc$;
     and (viii) $\cw \leq \mw + 2$.
\end{proposition}

\section{Background}

Before describing our algorithms, we introduce some building blocks for our argument.

\subsection{Entire Subgraphs}

Structural parameters such as modular-width and clique-width entail
the join operation in their underlying construction trees.
When joining two subgraphs in \PrbMaxHS,
it is important to distinguish whether all the vertices in the subgraph
are included in the happy set.
Formally, we introduce the notion of \textit{entire subgraphs}.

\begin{definition}
    Given a graph $G$ and a happy set $S$,
    an \textbf{entire subgraph} is a subgraph $G'$ of $G$ such that $V(G') \subseteq S$.
\end{definition}

By definition, the empty subgraph is always entire.
The following lemma is directly derived from the definition of happy vertices.

\begin{lemma}\label{lem:entire-unhappy}
    Let $G$ be a complete join of subgraphs $G_1$ and $G_2$.
    $V(G_1)$ admits a happy vertex only if $G_2$ is entire under a happy set $S \subseteq V(G)$.
\end{lemma}

\begin{proof}
    If $G_2$ is not entire, there must exist $v \in V(G_2)$ such that $v \notin S$.
    Recall Proposition~\ref{prop:unhappy}, and we have
    $N[V(G) \setminus S] \supseteq N(v) \supseteq V(G_1)$,
    which implies that any vertex in $V(G_1)$ cannot be happy.
\end{proof}

\subsection{Knapsack Variant with Non-linear Values}

The classic \textsc{Knapsack} problem has a number of variants,
 including \textsc{0-1 Knapsack} \cite{garey_1979_computers} and \textsc{Quadratic Knapsack} \cite{gallo_quadratic_1980}.
In this paper we consider another variant,
 where the objective function is the sum of non-linear functions,
 but the function range is limited to integers.
Specifically,
 each item has unit weight,
 but its value may vary depending on the number of copies of each type of item.
We also require the weight sum to be exact
 and call this problem \PrbFKnapsack, where $f$ stands for function.

\begin{problem}{\PrbFKnapsack}
    \Input & Given a set of $n$ items numbered from 1 to $n$,
             a weight capacity $W \in \Z_0^+$
             and 
             a value function $f_i: D_i \to \Z_0^+$,
             defined on a non-negative integral domain $D_i$ for each item $i$.\\
    \Prob & For every $1 \leq i \leq n$, find the number $x_i \in D_i$ of instances of item $i$ to include in the knapsack,
    maximizing $\sum_{i=1}^n f_i(x_i)$, subject to $\sum_{i=1}^n x_i = W$.
\end{problem}

We show that this problem is solvable in polynomial-time.

\begin{lemma}\label{lem:f-knapsack}
    \PrbFKnapsack can be solved in time $\bigo{nW^2}$.
\end{lemma}

\begin{proof}
    First, define the value $\phi[t,w]$ to be the maximum possible sum
     $\sum_{i=1}^t f_i(x_i)$,
     subject to $\sum_{i=1}^t x_i=w$ and $x_i \in D_i$ for every $i$.
    Then, perform bottom-up dynamic programming as follows.
    \begin{itemize}
        \item Initialize: $
        \phi[0,w]=\begin{cases}0 \quad\quad\text{if }w=0,\\
            -\infty \quad\text{if }w>0 \ (\text{meaning } \infeasible).
        \end{cases}$
        \item Update: $\displaystyle\phi[t,w] = \max_{x_t \in D_t \land x_t \leq w}f_t(x_t) + \phi[t-1,w-x_t]$
        \item Result: $\phi[n,w]$ for $0 \leq w \leq W$ is the optimal value for weight $w$.
    \end{itemize}

    The base case ($\phi[0,w]$) represents the state where no item is in the knapsack,
    so both the objective and weight are $0$; otherwise, infeasible.
    For the inductive step, any optimal solution $\phi[t,w]$ can be decomposed into
    $f_t(x_t) + \sum_{i=1}^{t-1}f_i(x_i)$ for some $x_t$,
    and the latter term ($\sum_{i=1}^{t-1}f_i(x_i)$) must equal $\phi[t-1,w-x_t]$ by definition.
    We consider all possible integers $x_t$, and thus the algorithm is correct.

    Since $0 \leq t \leq n$, $0 \leq w \leq W$, and the update takes time $\bigo{W}$,
    the total running time is $\bigo{nW^2}$.
    By using the standard technique of backlinks, one can reconstruct the solution $\{x_i\}$
    within the same asymptotic running time.
\end{proof}

The following result is a natural by-product of the algorithm above.

\begin{corollary}\label{cor:f-knapsack}
    Given an integer $W$, \PrbFKnapsack for all weight capacities $0 \leq w \leq W$ can be solved
    in total time $\bigo{nW^2}$.
\end{corollary}

\subsection{Integer Quadratic Programming}

For \PrbMaxEHS, we use the following known result that \PrbIQPLong is FPT by the number of variables
 and coefficients.

\begin{problem}{\PrbIQPLong (\PrbIQP)}
    \Input & An $n \times n$ integer matrix $Q$,
    an $m \times n$ integer matrix $A$
    and an $m$-dimensional integer vector $b$.\\
    \Prob & Find a vector $x \in \Z^n$ minimizing $x^T Q x$, subject to $Ax \leq b$.
\end{problem}

\begin{proposition}[Lokshtanov \cite{lokshtanov_parameterized_2017}]\label{prop:iqp}
    There exists an algorithm that given an instance of \PrbIQP,
    runs in time $f(n, \alpha)L^{\bigo{1}}$,
    and outputs a vector $x \in \Z^n$.
    If the input IQP has a feasible solution then $x$ is feasible,
    and if the input IQP is not unbounded, then $x$ is an optimal solution.
    Here $\alpha$ denotes the largest absolute value of an entry of $Q$ and $A$,
    and $L$ is the total number of bits required to encode the input.
\end{proposition}

It is convenient to have a linear term in the objective function.
This can be achieved by introducing a new variable $\hat{x} = 1$ and
 adding $[0, q]$ as the corresponding row in $Q$ \cite{lokshtanov_parameterized_2017}.

\begin{corollary}\label{col:iqp}
    Proposition~\ref{prop:iqp} holds if we generalize the objective function from
    $x^T Q x$ to $x^T Q x + q^T x$ for some $n$-dimensional integer vector $q$.
    Here $\alpha$ is the largest absolute value of an entry of $Q$, $q$ and $A$.
\end{corollary}

\section{Algorithms for Maximum Happy Set}

Now we describe our FPT algorithms for \PrbMaxHS with respect to modular-width and clique-width.
At a high level, we employ a bottom-up dynamic programming (DP) approach on the parse-tree of
a given graph, considering each node once.
At each node, we use several techniques on precomputed results to update the DP table.
For simplicity, our DP tables store the maximum number of happy vertices.
Like other DP applications, a \textit{certificate}, i.e. the actual happy set,
can be found by using backlinks within the same asymptotic running time.\looseness-1

\subsection{Parameterized by modular-width}

We give an algorithm whose running time is singly-exponential in
the modular-width, quadratic in $k$ and linear in the graph size.

\begin{theorem}\label{thm:hs-mw}
    \PrbMaxHS can be solved in time $\mathcal{O}(2^{\mw}\cdot \mw \cdot k^2 \cdot |V(G)|)$,
    where $\mw$ is the modular-width of the input graph $G$.
\end{theorem}

Our algorithm follows the common framework seen in \cite{modular-width-gajarsky}.
Given a graph $G$, a parse-tree with modular-width $\mw$ can be computed in linear-time \cite{tedder_simple_2007}.
The number of nodes in the parse-tree is linear in $|V(G)|$ \cite{modular-width-gajarsky}.
Our algorithm traverses the parse-tree from the bottom,
considering only operation (O4), as operations (O2)-(O3) can be replaced with a single operation (O4)
with at most two arguments \cite{modular-width-gajarsky}.
Further, we assume $2 \leq \mw < k$ without loss of generality.

Each node in the parse-tree corresponds to an induced subgraph of $G$, which we write $\mathcal{G}$.
We keep track of a table $\phi[\mathcal{G}, w]$,
the maximum number of happy vertices for $\mathcal{G}$ with regard to a happy set of size $w$.
We may assume $0 \leq w \leq k$ because we do not have to consider a happy set larger than size $k$.
The entries of the DP table are initialized with $\phi[\mathcal{G},w] = -\infty$.
For the base case, a graph $G_0$ with a single vertex introduced by operation (O1),
 we set $\phi[G_0,0]=0$ and $\phi[G_0,1]=1$.
The solution to the original problem is given by $\phi[G,k]$.

Our remaining task is to compute, given a graph substitution $\mathcal{G}=H(G_1,\ldots,G_n)$ ($n \leq \mw$),
the values of $\phi[\mathcal{G},w]$ provided partial solutions $\phi[G_1,\cdot],\ldots,\phi[G_n,\cdot]$.
We first choose a set of entire subgraphs from $G_1,\ldots,G_n$.
Then, we identify the \textit{subgraph type} for each $G_i$ during a graph substitution.

\begin{definition}[subgraph type]
    Given a graph substitution $H(G_1,\ldots,G_n)$, where $v_i \in V(H)$ is substituted by $G_i$,
     and a happy set $S$,
     we categorize each substituted subgraph $G_i$ into the following four types.

     \begin{itemize}
         \item Type I: $G_i$ is entire and for every $j$ such that $v_j \in N_H(v_i)$, $G_j$ is entire.
         \item Type II: $G_i$ is not entire and for every $j$ such that $v_j \in N_H(v_i)$, $G_j$ is entire.
         \item Type III: $G_i$ is entire and not Type I.
         \item Type IV: $G_i$ is not entire and not Type II.
     \end{itemize}
\end{definition}

Intuitively, Type I and II subgraphs are surrounded by entire subgraphs in $H$,
the \textit{metagraph} to substitute, and
Type I and III subgraphs are entire.
A pictorial representation of this partition is presented in \cref{fig:alg-hs-mw}.
Observe that from \cref{lem:entire-unhappy}, the subgraphs with Type III and IV cannot
 include any happy vertices.
Further, Type II subgraphs are independent in $H$ because their neighbors must be of Type III.
This ensures that the choice of a happy set in Type II is independent of other subgraphs.

\begin{figure*}[t]
    \centering

    \pgfdeclarelayer{bg}
    \pgfsetlayers{bg, main}
    \usetikzlibrary{calc}

    \tikzset{
        old inner xsep/.estore in=\oldinnerxsep,
        old inner ysep/.estore in=\oldinnerysep,
        double circle/.style 2 args={
            circle,
            old inner xsep=\pgfkeysvalueof{/pgf/inner xsep},
            old inner ysep=\pgfkeysvalueof{/pgf/inner ysep},
            /pgf/inner xsep=\oldinnerxsep+#1,
            /pgf/inner ysep=\oldinnerysep+#1,
            alias=sourcenode,
            append after command={
            let     \p1 = (sourcenode.center),
                    \p2 = (sourcenode.east),
                    \n1 = {\x2-\x1-#1-0.5*\pgflinewidth}
            in
                node [inner sep=0pt, draw, circle, minimum width=2*\n1,at=(\p1),#2] {}
            }
        },
        double circle/.default={-3pt}{black}
    }
    \tikzstyle{plain} = [circle, fill=white, text=black, draw, thick, scale=1, minimum size=0.5cm, inner sep=1.5pt]
    \tikzstyle{small} = [circle, fill=white, text=black, draw, thick, scale=1, minimum size=0.2cm, inner sep=1.5pt]
    \tikzstyle{entire} = [circle, fill=black!25, text=black, draw, thick, scale=1, minimum size=0.8cm, inner sep=1.5pt]
    \tikzstyle{large} = [circle, fill=white, text=black, draw, thick, scale=1, minimum size=0.8cm, inner sep=1.5pt]
    \tikzstyle{type1} = [double circle={-3pt}{black}, fill=black!25, text=black, draw, thick, scale=1, minimum size=0.8cm, inner sep=1.5pt]
    \tikzstyle{type2} = [double circle={-3pt}{black, dashed}, fill=white, text=black, draw, thick, scale=1, minimum size=0.8cm, inner sep=1.5pt]

    \begin{minipage}[m]{.96\textwidth}
        \vspace{0pt}
        \centering
        \begin{tikzpicture}
            \node[type1] (v1)  at (0, 2) {I};
            \node[type1] (v2)  at (0, 0) {I};
            \node[entire] (v3)  at (2, 2) {III};
            \node[type2]  (v4)  at (2, 0) {II};
            \node[entire] (v5)  at (4, 2) {III};
            \node[entire] (v6)  at (4, 0) {III};
            \node[large]  (v7)  at (6, 2) {IV};
            \node[type2]  (v8)  at (6, 0) {II};
            \node[large]  (v9)  at (8, 2) {IV};
            \node[type1] (v10) at (10, 2) {I};
            \node[type2]  (v11) at (10, 0) {II};
            \draw[very thick] (v1) -- (v2);
            \draw[very thick] (v1) -- (v3);
            \draw[very thick] (v2) -- (v3);
            \draw[very thick] (v3) -- (v4);
            \draw[very thick] (v3) -- (v5);
            \draw[very thick] (v4) -- (v5);
            \draw[very thick] (v4) -- (v6);
            \draw[very thick] (v5) -- (v7);
            \draw[very thick] (v6) -- (v8);
            \draw[very thick] (v7) -- (v9);

            \node at (-0.75, 2.4) {$G_1$};
            \node at (-0.75, 0.4) {$G_2$};
            \node at (1.4, 2.4) {$G_3$};
            \node at (1.25, 0.4) {$G_4$};
            \node at (3.4, 2.4) {$G_5$};
            \node at (3.4, 0.4) {$G_6$};
            \node at (5.4, 2.4) {$G_7$};
            \node at (5.3, 0.4) {$G_8$};
            \node at (7.4, 2.4) {$G_9$};
            \node at (9.2, 2.4) {$G_{10}$};
            \node at (9.2, 0.4) {$G_{11}$};
        \end{tikzpicture}
    \end{minipage}
    \caption{%
    Four types of the subgraphs after applying operation (O4).
    Entire subgraphs (Type I and III) are shaded in gray.
    A subgraph becomes Type III or IV if it has a non-entire neighbor (e.g. $G_3$, $G_9$).
    In a Type I subgraph, all vertices are happy.
    Type II subgraphs may or may not admit a happy vertex.
    Type III and IV subgraphs cannot contain a happy vertex, as it is adjacent to a non-entire subgraph.
    }
    \label{fig:alg-hs-mw}
\end{figure*}
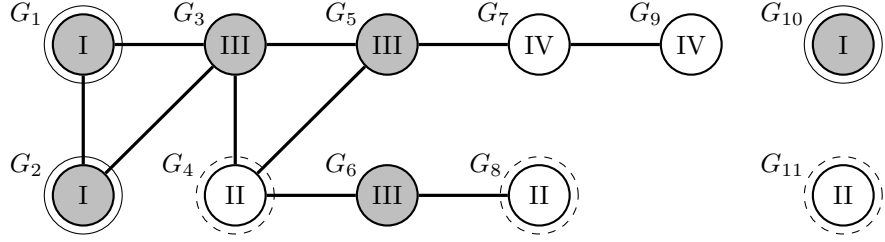

Lastly, we formulate an \PrbFKnapsack instance
as described in the following algorithm for updating the DP table on a single operation (O4).

\begin{alg}[\AlgMaxHSMW]
    \label{alg:max-hs-mw}
    Given a graph substitution $\mathcal{G}=H(G_1,\ldots,G_n)$ and 
    partial solutions $\phi[G_1,\cdot],\ldots,\phi[G_n,\cdot]$,
    consider all combinations of entire subgraphs from $G_1,\ldots,G_n$
    and proceed the following steps.

    (Step 1) Identify subgraph types for $G_1,\ldots,G_n$.
    (Step 2) Formulate an \PrbFKnapsack instance with capacity $k$ and value functions $f_i$,
    based on the subgraph $G_i$'s type as follows.
   \begin{itemize}
       \item \makebox[4em][l]{Type I}:   \makebox[8em][l]{$f_i(x) = \abs{G_i},$}  $x = \abs{G_i}$.
       \item \makebox[4em][l]{Type II}:  \makebox[8em][l]{$f_i(x) = \phi[G_i,x],$} $0 \leq x < \abs{G_i}$.
       \item \makebox[4em][l]{Type III}: \makebox[8em][l]{$f_i(x) = 0,$}          $x = \abs{G_i}$.
       \item \makebox[4em][l]{Type IV}:  \makebox[8em][l]{$f_i(x) = 0,$}          $0 \leq x < \abs{G_i}$.
   \end{itemize}
   Then, update the DP table entries $\phi[\mathcal{G}, w]$ for $0 \leq w \leq k$
    with the solution to \PrbFKnapsack,
    if its value is greater than the current value.
\end{alg}

We now prove that the runtime of this algorithm is FPT with respect to modular-width.

\begin{lemma}\label{lem:max-hs-mw}
    \cref{alg:max-hs-mw} correctly computes $\phi[\mathcal{G},w]$ for every $0 \leq w \leq k$
    in time $\bigo{2^{\mw}\cdot \mw \cdot k^2}$.
\end{lemma}

\begin{proof}
    First, the algorithm considers all possible sets of entire substituted subgraphs ($G_1,\ldots,G_n$).
    The optimal solution must belong to one of them.
    It remains to prove the correctness of the \PrbFKnapsack formulation in step 2.
    From \cref{lem:entire-unhappy}, the subgraphs of Type III and IV cannot increase the number of
     happy vertices, so we set $f_i(x)=0$.
    For Type I, the algorithm has no option but
     to include all of $V(G_i)$ in the happy set, and they are all happy.

    The subgraphs of Type II are the only ones that use previous results, $\phi[G_i,\cdot]$.
    Since the new neighbors to $G_i$ are required to be in the happy set, for any choice of the happy set in $G_i$,
    happy vertices remain happy, and unhappy vertices remain unhappy.
    Thus, we can directly use $\phi[G_i,\cdot]$;
     its choice does not affect other substituted subgraphs,
     as Type II subgraphs are independent in $H$.
    The domain of functions $f_i$ is naturally determined by the definition of entire subgraphs.

    Now, consider the running time of \cref{alg:max-hs-mw}.
    It considers $2^n$ possible combinations of entire subgraphs.
    Step 1 can be done by checking neighbors for each vertex in $H$,
    so the running time is $\bigo{\abs{E(H)}} = \bigo{n ^ 2}$.
    And step 2 takes time $\bigo{n k^2}$ from \cref{cor:f-knapsack}.
    The total running time is $\bigo{2^n (n^2 + nk^2)} = \bigo{2^{\mw}\cdot \mw \cdot k^2}$
     as we assume $n \leq \mw < k$.
\end{proof}

\begin{proof}[Proof of \cref{thm:hs-mw}]
    It is trivial to see that the base case of the DP is valid,
    and the correctness of inductive steps is given by \cref{lem:max-hs-mw}.
    We process each node of the parse-tree once,
     and it has $\bigo{\abs{V(G)}}$ nodes \cite{modular-width-gajarsky}.
    Thus, the overall runtime is $\mathcal{O}(2^{\mw}\cdot \mw \cdot k^2 \cdot |V(G)|)$.
\end{proof}

\subsection{Parameterized by clique-width}

We provide an algorithm for \PrbMaxHS parameterized only by clique-width ($\cw$),
 which no longer requires a combined parameter with solution size $k$
 as presented in \cite{asahiro2021parameterized}.

\begin{theorem}\label{thm:hs-cw}
    Given a $\cw$-expression tree of a graph $G$ with clique-width $\cw$,
    \PrbMaxHS can be solved in time $\mathcal{O}(8^{\cw} \cdot k^2 \cdot |V(G)|)$.
\end{theorem}

Here we assume that we are given a $\cw$-expression tree,
where each node $t$ represents a \textit{labeled} graph $G_t$.
A labeled graph is a graph whose vertices are labeled by integers in $L=\{1, \ldots, \cw\}$.
Every node must be one of the following:
introduce node $i(v)$,
union node $G_1 \oplus G_2$,
relabel node $\rho(i,j)$, or
join node $\eta(i,j)$.
We write $V_i$ for the set of vertices with label $i$.\looseness-1

Our algorithm traverses the $\cw$-expression tree from the leaves
 and performs dynamic programming.
For every node $t$, we keep track of the annotated partial solution $\phi[t, w, X, T]$,
for every integer $0 \leq w \leq k$ and sets of labels $X, T \subseteq L$.
We call $X$ the \textit{entire labels} and $T$ the \textit{target labels}.
$\phi[t, w, X, T]$ is defined to be the maximum number of happy vertices
 having target labels $T$ for $G_t$
 with respect to a happy set $S \subseteq V(G_t)$ of size $w$ such that
 $V_\ell$ is entire to $S$ if and only if $\ell \in X$.
The entries of the DP table are initialized with $\phi[t, w, X, T] = -\infty$.
The solution to the original graph $G$ is computed by $\max_{X \subseteq L} \phi[r, k, X, L]$,
where $r$ is the root of the $\cw$-expression tree.
Now we claim the following recursive formula for each node type.\looseness-1

\begin{lemma}[Formula for introduce nodes]
    Suppose $t$ is an introduce node, where a vertex $v$ with label $i$ is introduced.
    Then, the following holds.
    \begin{align*}
        \phi[t,w,X,T] &= \begin{cases}
            1 &\text{if } w=1,\  X=L \text{ and } i \in T;\\
            0 &\text{if } w=1,\  X=L \text{ and } i \notin T;\\
            0 &\text{if } w=0 \text{ and } X = L \setminus \{i\};\\
            -\infty &\text{otherwise.}
        \end{cases}
    \end{align*}
\end{lemma}

\begin{proof}
    First, notice that all labels but $i$ are empty and thus entire.
    If we include $v$ in the happy set, then we get $w=1$ and $X=L$ (all labels are entire).
    The resulting value depends on the target labels.
    If $i$ is a target label, i.e. $i \in T$, then $v$ is a happy vertex having a target label,
     resulting in $\phi[t,w,X,T]=1$.
    Otherwise, $\phi[t,w,X,T]=0$.
    If $w=0$, then label $i$ cannot be entire, and the only feasible solution is
     $\phi[t,0,L \setminus \{i\},T]=0$.
\end{proof}

\begin{lemma}[Formula for union nodes]
    Suppose $t$ is a union node with child nodes $t_1$ and $t_2$.
    Then, the following holds.
    \begin{align*}
        \phi[t, w, X, T] &=
        \max_{0 \leq \tilde{w} \leq w}
        \max_{\substack{X_1, X_2 \subseteq L\\: X_1 \cap X_2 = X}}
        \phi[t_1, \tilde{w}, X_1, T] + \phi[t_2, w-\tilde{w}, X_2, T]
    \end{align*}
\end{lemma}

\begin{proof}
    At a union node, since $G_{t_1}$ and $G_{t_2}$ are disjoint,
     any maximum happy set in $G_t$ must be the disjoint union of
     some maximum happy set in $G_{t_1}$ and that in $G_{t_2}$ for the same target labels.
    We consider all possible combinations of partial solutions to $G_{t_1}$ and $G_{t_2}$,
     so the optimality is preserved.
    Note that a label in $G_t$ is entire if and only if it is entire in both $G_{t_1}$ and $G_{t_2}$.\looseness-1
\end{proof}

\begin{lemma}[Formula for relabel nodes]\label{lem:cw-relabel}
    Suppose $t$ is a relabel node with child node $t'$,
    where label $i$ in graph $G_{t'}$ is relabeled to $j$.
    Then, the following holds.
    \begin{align*}
        \phi[t, w, X, T] &= \begin{cases}
            -\infty &\text{if } i \notin X;\\
            \phi[t', w, X, T'] &\text{if } i \in X \text{ and } j \in X;\\
            \displaystyle\max_{Y \in \{\emptyset, \{i\}, \{j\}\}} \phi[t', w, X \setminus \{i\} \cup Y, T']
                &\text{if } i \in X \text{ and } j \notin X,\\
        \end{cases}
    \end{align*}
    where $T' = T \cup \{i\}$ if $j \in T$ and $T \setminus \{i\}$ otherwise.
\end{lemma}

\begin{proof}
    At a relabel node $\rho(i,j)$, label $i$ becomes empty, so it must be entire in $G_t$,
    leading to the first case.
    The variable $T'$ converts the target labels in $G_{t'}$ to those in $G_t$.
    If label $j$ is a target in $G_t$, then $i$ and $j$ must be targets in $G_{t'}$.
    Likewise, if label $j$ is not a target in $G_t$, then neither $i$ nor $j$ should be targets in $G_{t'}$.

    If label $j$ is entire in $G_t$,
     then the maximum happy set must be the same as the one in $G_{t'}$
     where both labels $i$ and $j$ are entire.
    If $j$ is not entire in $G_t$, then we need to choose the best solution from the following:
    $i$ is entire but $j$ is not, $j$ is entire but $i$ is not, neither $i$ nor $j$ is entire.
    Because $G_t$ and $G_{t'}$ have the same underlying graph, the optimal solution must be one of these.
\end{proof}

\begin{figure*}[t]
    \centering

    \pgfdeclarelayer{bg}
    \pgfsetlayers{bg, main}
    \usetikzlibrary{calc}

    \tikzset{
        old inner xsep/.estore in=\oldinnerxsep,
        old inner ysep/.estore in=\oldinnerysep,
        double circle/.style 2 args={
            circle,
            old inner xsep=\pgfkeysvalueof{/pgf/inner xsep},
            old inner ysep=\pgfkeysvalueof{/pgf/inner ysep},
            /pgf/inner xsep=\oldinnerxsep+#1,
            /pgf/inner ysep=\oldinnerysep+#1,
            alias=sourcenode,
            append after command={
            let     \p1 = (sourcenode.center),
                    \p2 = (sourcenode.east),
                    \n1 = {\x2-\x1-#1-0.5*\pgflinewidth}
            in
                node [inner sep=0pt, draw, circle, minimum width=2*\n1,at=(\p1),#2] {}
            }
        },
        double circle/.default={-3pt}{black}
    }
    \tikzstyle{large} = [circle, fill=white, text=black, draw, thick, scale=1, minimum size=0.8cm, inner sep=1.5pt]
    \tikzstyle{type1} = [double circle={-3pt}{black}, fill=black!25, text=black, draw, thick, scale=1, minimum size=0.8cm, inner sep=1.5pt]
    \tikzstyle{type2} = [double circle={-3pt}{black, dashed}, fill=white, text=black, draw, thick, scale=1, minimum size=0.8cm, inner sep=1.5pt]
    \tikzstyle{plain} = [circle, fill=white, text=black, draw, thick, scale=1, minimum size=0.3cm, inner sep=1.5pt]
    \tikzstyle{happy} = [double circle={-3.5pt}{line width=1.2pt}, fill=black!25, text=black, draw, thick, scale=1, minimum size=0.3cm, inner sep=1.5pt]
    \tikzstyle{unhappy} = [circle, fill=black!25, text=black, draw, thick, scale=1, minimum size=0.3cm, inner sep=1.5pt]

    \begin{minipage}[m]{.96\textwidth}
        \vspace{0pt}
        \centering
        \begin{tikzpicture}
            \node[happy] (a1) at (5, 12) {};
            \node[unhappy] (b1) at (5, 11) {};
            \node[plain] (c1) at (5, 10) {};
            \node[happy] (d1) at (6.5, 12) {};
            \node[happy] (e1) at (6.5, 11) {};
            \node[plain] (f1) at (6.5, 10) {};

            \node[unhappy] (a2) at (9, 12) {};
            \node[unhappy] (b2) at (9, 11) {};
            \node[happy] (c2) at (9, 10) {};
            \node[plain] (d2) at (10.5, 12) {};
            \node[plain] (e2) at (10.5, 11) {};
            \node[happy] (f2) at (10.5, 10) {};

            \node[unhappy] (a3) at (13, 12) {};
            \node[plain] (b3) at (13, 11) {};
            \node[plain] (c3) at (13, 10) {};
            \node[happy] (d3) at (14.5, 12) {};
            \node[unhappy] (e3) at (14.5, 11) {};
            \node[unhappy] (f3) at (14.5, 10) {};

            \node[happy] (a4) at (5, 8) {};
            \node[unhappy] (b4) at (5, 7) {};
            \node[plain] (c4) at (5, 6) {};
            \node[happy] (d4) at (6.5, 8) {};
            \node[happy] (e4) at (6.5, 7) {};
            \node[plain] (f4) at (6.5, 6) {};

            \draw[] (a1) -- (b1);
            \draw[] (b1) -- (c1);
            \draw[] (a1) -- (d1);
            \draw[] (b1) -- (e1);
            \draw[] (b1) -- (f1);
            \draw[] (c1) -- (f1);
            \draw[] (d1) -- (e1);

            \draw[] (a2) -- (b2);
            \draw[] (b2) -- (c2);
            \draw[] (a2) -- (d2);
            \draw[] (b2) -- (e2);
            \draw[] (b2) -- (f2);
            \draw[] (c2) -- (f2);
            \draw[] (d2) -- (e2);

            \draw[] (a3) -- (b3);
            \draw[] (b3) -- (c3);
            \draw[] (a3) -- (d3);
            \draw[] (b3) -- (e3);
            \draw[] (b3) -- (f3);
            \draw[] (c3) -- (f3);
            \draw[] (d3) -- (e3);

            \draw[] (a4) -- (b4);
            \draw[] (b4) -- (c4);
            \draw[] (a4) -- (d4);
            \draw[] (b4) -- (e4);
            \draw[] (b4) -- (f4);
            \draw[] (c4) -- (f4);
            \draw[] (d4) -- (e4);

            \node at (3, 11) {$G_{t'}$};
            \node at (3.5, 9) {$\rho(i,j)$};
            \node at (3, 7) {$G_t$};
            \draw[-{Latex[length=2mm, width=1.5mm]}] (3, 10.5)--(3, 7.5);
            \draw[-{Latex[length=2mm, width=1.5mm]}] (5.75, 9.4)--(5.75, 8.6);

            \draw[rounded corners, dashed] (4.6, 9.6) rectangle ++ (0.8, 2.8);
            \draw[rounded corners, dashed] (6.1, 9.6) rectangle ++ (0.8, 2.8);
            \draw[] (4.0, 9.4) rectangle ++ (3.5, 3.2);
            \node at (4.3, 12.2) {$i$};
            \node at (7.2, 12.2) {$j$};
            \node at (5.8, 13) {$\phi[t', 4, \emptyset, \{i,j\}] = \textbf{3}$};

            \draw[rounded corners, dashed] (8.6, 9.6) rectangle ++ (0.8, 2.8);
            \draw[rounded corners, dashed] (10.1, 9.6) rectangle ++ (0.8, 2.8);
            \draw[] (8.0, 9.4) rectangle ++ (3.5, 3.2);
            \node at (8.3, 12.2) {$i$};
            \node at (11.2, 12.2) {$j$};
            \node at (9.8, 13) {$\phi[t', 4, \{i\}, \{i,j\}] = 2$};

            \draw[rounded corners, dashed] (12.6, 9.6) rectangle ++ (0.8, 2.8);
            \draw[rounded corners, dashed] (14.1, 9.6) rectangle ++ (0.8, 2.8);
            \draw[] (12.0, 9.4) rectangle ++ (3.5, 3.2);
            \node at (12.3, 12.2) {$i$};
            \node at (15.2, 12.2) {$j$};
            \node at (13.8, 13) {$\phi[t', 4, \{j\}, \{i,j\}] = 1$};

            \draw[rounded corners, dashed] (4.6, 5.6) rectangle ++ (2.3, 2.8);
            \node at (7.2, 8.2) {$j$};
            \draw[] (4.0, 5.4) rectangle ++ (3.5, 3.2);

            \node at (9, 7) {$\phi[t, 4, \{i\}, \{j\}] = 3$};
            \node[align=left] at (6.7, 5.2) {entire: $i$};

            \node[align=left] at (6.7, 9.2) {entire: $\emptyset$};
            \node[align=left] at (10.7, 9.2) {entire: $i$};
            \node[align=left] at (14.7, 9.2) {entire: $j$};
        \end{tikzpicture}
    \end{minipage}
    \caption{%
    Visualization of the relabel operation $\rho(i,j)$ in the $\cw$-expression tree of labels $i,j$.
    Figures show happy sets of size $4$ (shaded in gray) maximizing the number of happy vertices
    (shown with double circles).
    After relabeling $i$ to $j$, label $i$ becomes empty and thus entire.
    Since label $j$ in $G_t$ corresponds to labels $i$ and $j$ in $G_{t'}$,
    to compute $\phi[t,4,\{i\},\{j\}]$,
    we need to look up three partial solutions
     $\phi[t',4,\emptyset,T']$, $\phi[t',4,\{i\},T']$, and $\phi[t',4,\{j\},T']$,
     where $T'=\{i,j\}$,
     and keep the one with the largest value (the left one in this example).
    }
    \label{fig:alg-hs-cw-relabel}
\end{figure*}
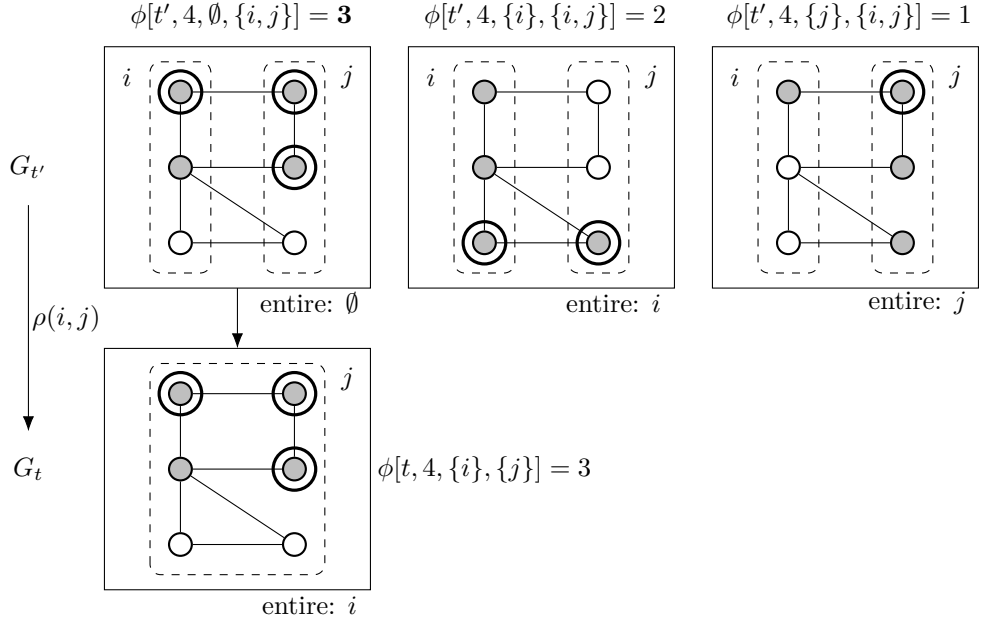

\cref{fig:alg-hs-cw-relabel} illustrates the third case of the formula in \cref{lem:cw-relabel}.

\begin{lemma}[Formula for join nodes]\label{lem:cw-join}
    Suppose $t$ is a join node with the child node $t'$, where labels $i$ and $j$ are joined.
    Then, the following holds.
    \begin{align*}
        \phi[t, w, X, T] &= \begin{cases}
            \phi[t', w, X, T] &\text{if } i \in X \text{ and } j \in X\\
            \phi[t', w, X, T \setminus \{i\}] &\text{if } i \in X \text{ and } j \notin X\\
            \phi[t', w, X, T \setminus \{j\}] &\text{if } i \notin X \text{ and } j \in X\\
            \phi[t', w, X, T \setminus \{i,j\}] &\text{if } i \notin X \text{ and } j \notin X\\
        \end{cases}
    \end{align*}
\end{lemma}

\begin{proof}
    At a join node $\eta(i,j)$, first observe that for any happy set,
    the vertices labeled other than $i,j$ are unaffected;
    happy vertices remain happy.
   Further, if label $j$ is not entire in $G_t$,
    then all vertices in $V_i$ cannot be happy from \cref{lem:entire-unhappy}.
   Thus, the maximum happy set in $G_t$ is equivalent to
    the one in $G_{t'}$ such that label $i$ is not a target label.
   The same argument applies to the other cases.
\end{proof}

\begin{figure*}[t]
    \centering

    \pgfdeclarelayer{bg}
    \pgfsetlayers{bg, main}
    \usetikzlibrary{calc}

    \tikzset{
        old inner xsep/.estore in=\oldinnerxsep,
        old inner ysep/.estore in=\oldinnerysep,
        double circle/.style 2 args={
            circle,
            old inner xsep=\pgfkeysvalueof{/pgf/inner xsep},
            old inner ysep=\pgfkeysvalueof{/pgf/inner ysep},
            /pgf/inner xsep=\oldinnerxsep+#1,
            /pgf/inner ysep=\oldinnerysep+#1,
            alias=sourcenode,
            append after command={
            let     \p1 = (sourcenode.center),
                    \p2 = (sourcenode.east),
                    \n1 = {\x2-\x1-#1-0.5*\pgflinewidth}
            in
                node [inner sep=0pt, draw, circle, minimum width=2*\n1,at=(\p1),#2] {}
            }
        },
        double circle/.default={-3pt}{black}
    }
    \tikzstyle{large} = [circle, fill=white, text=black, draw, thick, scale=1, minimum size=0.8cm, inner sep=1.5pt]
    \tikzstyle{type1} = [double circle={-3pt}{black}, fill=black!25, text=black, draw, thick, scale=1, minimum size=0.8cm, inner sep=1.5pt]
    \tikzstyle{type2} = [double circle={-3pt}{black, dashed}, fill=white, text=black, draw, thick, scale=1, minimum size=0.8cm, inner sep=1.5pt]
    \tikzstyle{plain} = [circle, fill=white, text=black, draw, thick, scale=1, minimum size=0.3cm, inner sep=1.5pt]
    \tikzstyle{happy} = [double circle={-3.5pt}{line width=1.2pt}, fill=black!25, text=black, draw, thick, scale=1, minimum size=0.3cm, inner sep=1.5pt]
    \tikzstyle{unhappy} = [circle, fill=black!25, text=black, draw, thick, scale=1, minimum size=0.3cm, inner sep=1.5pt]

    \begin{minipage}[m]{.96\textwidth}
        \vspace{0pt}
        \centering
        \begin{tikzpicture}
            \node[happy] (a1) at (5, 12) {};
            \node[happy] (b1) at (5, 11) {};
            \node[happy] (c1) at (5, 10) {};
            \node[happy] (d1) at (6.5, 12) {};
            \node[plain] (e1) at (6.5, 11) {};
            \node[unhappy] (f1) at (6.5, 10) {};

            \node[unhappy] (a2) at (9, 12) {};
            \node[unhappy] (b2) at (9, 11) {};
            \node[happy] (c2) at (9, 10) {};
            \node[plain] (d2) at (10.5, 12) {};
            \node[happy] (e2) at (10.5, 11) {};
            \node[happy] (f2) at (10.5, 10) {};

            \node[unhappy] (a4) at (9, 8) {};
            \node[unhappy] (b4) at (9, 7) {};
            \node[unhappy] (c4) at (9, 6) {};
            \node[plain] (d4) at (10.5, 8) {};
            \node[happy] (e4) at (10.5, 7) {};
            \node[happy] (f4) at (10.5, 6) {};

            \draw[] (a1) -- (b1);
            \draw[] (b1) -- (c1);
            \draw[] (a1) -- (d1);
            \draw[] (b1) -- (d1);
            \draw[] (c1) -- (f1);
            \draw[] (e1) -- (f1);
            \draw[] (c1) -- (f1);

            \draw[] (a2) -- (b2);
            \draw[] (b2) -- (c2);
            \draw[] (a2) -- (d2);
            \draw[] (b2) -- (d2);
            \draw[] (c2) -- (f2);
            \draw[] (e2) -- (f2);
            \draw[] (c2) -- (f2);

            \draw[] (a4) -- (b4);
            \draw[] (b4) -- (c4);
            \draw[] (a4) -- (d4);
            \draw[] (b4) -- (d4);
            \draw[] (c4) -- (d4);
            \draw[] (e4) -- (f4);
            \draw[] (a4) -- (e4);
            \draw[] (b4) -- (e4);
            \draw[] (c4) -- (e4);
            \draw[] (a4) -- (f4);
            \draw[] (b4) -- (f4);
            \draw[] (c4) -- (f4);

            \node at (3, 11) {$G_{t'}$};
            \node at (3.5, 9) {$\eta(i,j)$};
            \node at (3, 7) {$G_t$};
            \draw[-{Latex[length=2mm, width=1.5mm]}] (3, 10.5)--(3, 7.5);
            \draw[-{Latex[length=2mm, width=1.5mm]}] (9.75, 9.4)--(9.75, 8.6);

            \draw[rounded corners, dashed] (4.6, 9.6) rectangle ++ (0.8, 2.8);
            \draw[rounded corners, dashed] (6.1, 9.6) rectangle ++ (0.8, 2.8);
            \draw[] (4.0, 9.4) rectangle ++ (3.5, 3.2);
            \node at (4.3, 12.2) {$i$};
            \node at (7.2, 12.2) {$j$};
            \node at (5.8, 13) {$\phi[t', 5, \{i\}, \{i,j\}] = 4$};

            \draw[rounded corners, dashed] (8.6, 9.6) rectangle ++ (0.8, 2.8);
            \draw[rounded corners, dashed] (10.1, 9.6) rectangle ++ (0.8, 2.8);
            \draw[] (8.0, 9.4) rectangle ++ (3.5, 3.2);
            \node at (8.3, 12.2) {$i$};
            \node at (11.2, 12.2) {$j$};
            \node at (9.8, 13) {$\phi[t', 5, \{i\}, \{j\}] = 2$};

            \draw[rounded corners, dashed] (8.6, 5.6) rectangle ++ (0.8, 2.8);
            \draw[rounded corners, dashed] (10.1, 5.6) rectangle ++ (0.8, 2.8);
            \node at (8.3, 8.2) {$i$};
            \node at (11.2, 8.2) {$j$};
            \draw[] (8.0, 5.4) rectangle ++ (3.5, 3.2);

            \node at (13.2, 7) {$\phi[t, 5, \{i\}, \{i,j\}] = 2$};
            \node[align=left] at (9.75, 5.2) {entire: $i$, target: $i,j$};

            \node[align=left] at (5.75, 9.2) {entire: $i$, target: $i,j$};
            \node[align=left] at (9.75, 9.2) {entire: $i$, target: $j$};
        \end{tikzpicture}
    \end{minipage}
    \caption{%
    Visualization of the join operation $\eta(i,j)$ in the $\cw$-expression tree of labels $i,j$.
    Figures show happy sets of size $5$ (shaded in gray) maximizing the number of happy vertices
    (shown with double circles) for different target labels.
    We consider the case where label $i$ is entire and $j$ is not.
    The graph $G_{t'}$ admits $4$ happy vertices if both $i$ and $j$ are target labels.
    However, this is no longer true after the join because label $j$ is not entire.
    Instead, the optimal happy set for $G_t$ can be found where
     label $i$ is excluded from the target labels for $G_{t'}$, i.e. $\phi[t',5,\{i\},\{j\}]$,
     which admits $2$ happy vertices with label $j$.
    Notice that we do not count the happy vertices with label $i$ if it is not a target.
    }
    \label{fig:alg-hs-cw-join}
\end{figure*}
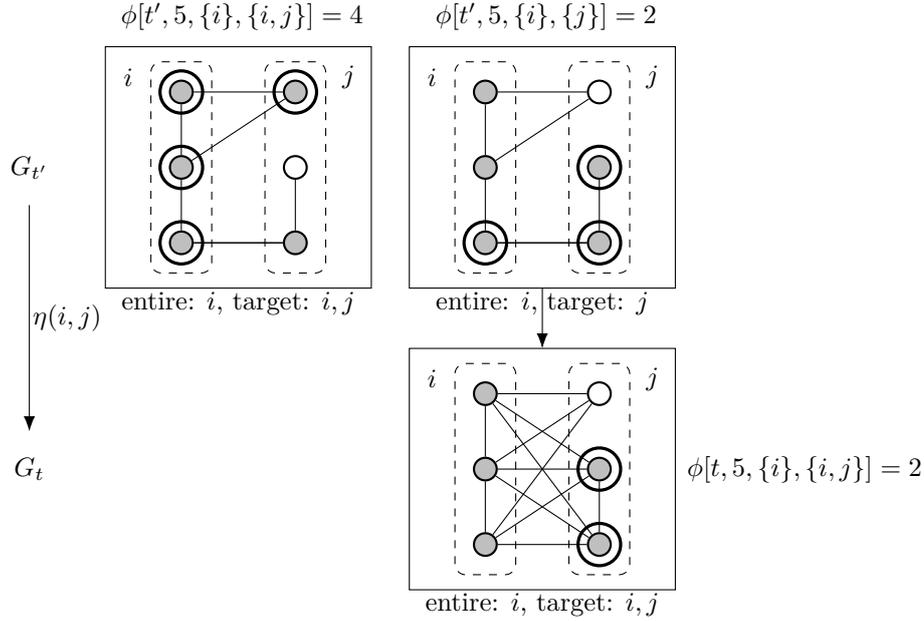

\cref{fig:alg-hs-cw-join} illustrates the second case of the formula in \cref{lem:cw-join}.
Lastly, we examine the running time of these computations.

\begin{proposition}
    Given a $\cw$-expression tree and its node $t$,
     and partial solutions $\phi[t',\cdot,\cdot,\cdot]$ for every child node $t'$ of $t$,
     we can compute $\phi[t, w, X, T]$ for every $w,X,T$ in time $\bigo{8^{\cw}\cdot k^2}$.
\end{proposition}

\begin{proof}
    It is clear to see that for fixed $t, w, X, T$, the formulae for introduce, relabel, and join nodes
     take $\bigo{1}$.
    If we compute $\phi[t, \cdot,\cdot,\cdot]$ for every $w, X, T$, the total running time is
     $\bigo{\left(2^{\cw}\right)^2\cdot k}$ since $w$ is bounded by $k$ and there are $2^{\cw}$ configurations for $X$ and $T$.

    For the union node formula, observe that $X$ can be determined by the choice of $X_1$ and $X_2$,
     so it is enough to consider all possible values for $w,T,\tilde{w},X_1,X_2$,
     which results in the running time $\bigo{\left(2^{\cw}\right)^3\cdot k^2}$,
     or $\bigo{8^{\cw}\cdot k^2}$.
\end{proof}

This completes the proof of \cref{thm:hs-cw}, as we process each node of the $\cw$-expression tree once,
 and it has $\bigo{\abs{V(G)}}$ nodes.

\section{Algorithms for Maximum Edge Happy Set}\label{sec:algorithm2}

In addition to \PrbMaxHS, we also study its edge-variant \PrbMaxEHS.
One difference from \PrbMaxHS is that when joining two subgraphs,
 we may increase the number of edges between those subgraphs,
 even if they are not entire.
In other words, the number of edges between joining subgraphs depends on two variables,
 and quadratic programming naturally takes part in this setting.
Here, we present FPT algorithms for two parameters---%
 neighborhood diversity and cluster deletion number---%
 to figure out the boundary between parameters that are FPT
 (e.g. treewidth) and W[1]-hard (e.g. clique-width) (see \cref{tab:result}).
%

\subsection{Parameterized by neighborhood diversity}

As shown in \cref{fig:params}, neighborhood diversity is a parameter specializing modular-width.
To obtain a finer classification of structural parameters,
 we now show \PrbMaxEHS is FPT parameterized by neighborhood diversity.

Let $\nd$ be the neighborhood diversity of the given graph $G$.
We observe that any instance $(G, k)$ of \PrbMaxEHS can be reduced to the instance of \PrbIQPLong (\PrbIQP)
 as follows.

\begin{lemma}\label{lem:ehs-nd-reduction}
    \PrbMaxEHS can be reduced to \PrbIQP with $\bigo{\nd}$ variables and bounded coefficients
    in time $\bigo{\abs{V(G)}+\abs{E(G)}}$.
\end{lemma}

\begin{proof}
    First, we compute the set of twins (modules) $\mathcal{M} = M_1, \ldots, M_{\nd}$ of $G$,
     and obtain the quotient graph $H$
     on the modules $\mathcal{M}$ in time $\bigo{\abs{V(G)}+\abs{E(G)}}$ \cite{lampis_algorithmic_2012}.
    Note that each module $M_i$ is either a clique or an independent set.
    Let us define a vector $q \in \Z^{\nd}$ such that
    $q_i = 1$ if $M_i$ is a clique, and $q_i=0$ if $M_i$ is an independent set.
    Further, let $A \in \Z^{\nd \times \nd}$ be the adjacency matrix of $H$
     where $A_{ij}=1$ if $M_i M_j \in E(H)$, and $A_{ij}=0$ otherwise.
    
    \br
    We then formulate an \PrbIQP instance as follows:

    \begin{itemize}
        \item Variables: $x \in \Z^\nd$.
        \item Maximize: $f(x) = x^T (A + qq^T) x  - q^T x$ (equivalently, minimize $-f(x)$).
        \item Subject to: $\sum_i x_i = k$ and
        $0 \leq x_i \leq |M_i|$ for every $1 \leq i \leq \nd$.
    \end{itemize}

    This formulation has $\nd$ variables,
    and its coefficients are either $0$ or $\pm 1$, thus bounded.
    After finding the optimal vector $x$, pick any $x_i$ vertices from module $M_i$
    and include them in the happy set $S$.
    We claim that $S$ maximizes the number of happy edges.

    For any happy set $S$, the number of happy edges is given by the sum of
     the number of happy edges inside each module $M_i$, which we call \textit{internal edges},
     and the number of edges between each module pair $M_i$ and $M_j$, or \textit{external edges}.
    Let $x \in \Z^{\nd}$ be a vector such that $x_i = \abs{S \cap M_i}$ for every $i$.
    Then, the number of internal edges of module $M_i$ is $q_i \cdot \binom{x_i}{2}$,
    and the number of external edges between modules $M_i$ and $M_j$ is $A_{ij} \cdot x_i x_j$.
    The number of happy edges, i.e. $\abs{E(G[S])}$, is given by:
    $h(x) = \left[\sum_{i=1}^{\nd} q_i \cdot \binom{x_i}{2}\right] 
    + \left[\sum_{1 \leq i < j \leq \nd} A_{ij}\cdot x_i x_j\right]$.
    One can trivially verify $f(x) = 2h(x)$.

    If the \PrbIQP instance is feasible, then we can find a happy set $S$ of size $k$ maximizing $h(x)$,
     which must be the optimal solution to \PrbMaxEHS.
    Otherwise, $\sum_{i} \abs{M_i} = \abs{V(G)} < k$, and \PrbMaxEHS is also infeasible.
\end{proof}

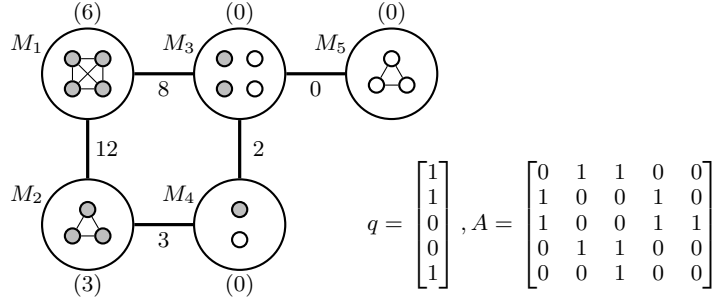
\begin{figure}[t]
    \centering
    \pgfdeclarelayer{bg}
    \pgfsetlayers{bg, main}
    \usetikzlibrary{calc}

    \tikzset{
        old inner xsep/.estore in=\oldinnerxsep,
        old inner ysep/.estore in=\oldinnerysep,
        double circle/.style 2 args={
            circle,
            old inner xsep=\pgfkeysvalueof{/pgf/inner xsep},
            old inner ysep=\pgfkeysvalueof{/pgf/inner ysep},
            /pgf/inner xsep=\oldinnerxsep+#1,
            /pgf/inner ysep=\oldinnerysep+#1,
            alias=sourcenode,
            append after command={
            let     \p1 = (sourcenode.center),
                    \p2 = (sourcenode.east),
                    \n1 = {\x2-\x1-#1-0.5*\pgflinewidth}
            in
                node [inner sep=0pt, draw, circle, minimum width=2*\n1,at=(\p1),#2] {}
            }
        },
        double circle/.default={-3pt}{black}
    }
    \tikzstyle{small} = [circle, fill=white, text=black, draw, thick, scale=1, minimum size=0.2cm, inner sep=1.5pt]
    \tikzstyle{happy} = [circle, fill=black!25, text=black, draw, thick, scale=1, minimum size=0.2cm, inner sep=1.5pt]
    \tikzstyle{large} = [circle, fill=white, text=black, draw, thick, scale=1, minimum size=1.2cm, inner sep=1.5pt]
    \begin{minipage}{0.96\textwidth}
        \centering
        \footnotesize
        \begin{tikzpicture}
            \node[large] (m1) at (0, 2) {};
            \node[large] (m2) at (0, 0) {};
            \node[large] (m3) at (2, 2) {};
            \node[large] (m4) at (2, 0) {};
            \node[large] (m5) at (4, 2) {};
   
            \node[happy] (v11) at (-0.2, 1.8) {};
            \node[happy] (v12) at (-0.2, 2.2) {};
            \node[happy] (v13) at (0.2, 1.8) {};
            \node[happy] (v14) at (0.2, 2.2) {};

            \node[happy] (v21) at (0, 0.2) {};
            \node[happy] (v22) at (-0.2, -0.15) {};
            \node[happy] (v23) at (0.2, -0.15) {};
            
            \node[happy] (v31) at (1.8, 1.8) {};
            \node[happy] (v32) at (1.8, 2.2) {};
            \node[small] (v33) at (2.2, 1.8) {};
            \node[small] (v34) at (2.2, 2.2) {};

            \node[happy] (v41) at (2, 0.2) {};
            \node[small] (v42) at (2, -0.2) {};

            \node[small] (v51) at (4, 2.2) {};
            \node[small] (v52) at (3.8, 1.85) {};
            \node[small] (v53) at (4.2, 1.85) {};
            
            \draw[very thick] (m1) -- (m2);
            \draw[very thick] (m1) -- (m3);
            \draw[very thick] (m3) -- (m4);
            \draw[very thick] (m2) -- (m4);
            \draw[very thick] (m3) -- (m5);

            \draw (v11) -- (v12);
            \draw (v11) -- (v13);
            \draw (v11) -- (v14);
            \draw (v12) -- (v13);
            \draw (v12) -- (v14);
            \draw (v13) -- (v14);

            \draw (v21) -- (v22);
            \draw (v21) -- (v23);
            \draw (v22) -- (v23);

            \draw (v51) -- (v52);
            \draw (v51) -- (v53);
            \draw (v52) -- (v53);

            \node at (-0.8, 2.4) {$M_1$};
            \node at (-0.8, 0.4) {$M_2$};

            \node at (1.2, 2.4) {$M_3$};
            \node at (1.2, 0.4) {$M_4$};
            
            \node at (3.2, 2.4) {$M_5$};
            
            \node at (0.25, 1) {$12$};
            \node at (2.25, 1) {$2$};
            \node at (1, -0.2) {$3$};
            \node at (1, 1.8) {$8$};
            \node at (3, 1.8) {$0$};

            \node at (0, 2.8) {$(6)$};
            \node at (0, -0.8) {$(3)$};
            \node at (2, -0.8) {$(0)$};
            \node at (2, 2.8) {$(0)$};
            \node at (4, 2.8) {$(0)$};

            \node at (6, 0) {$
            q = \begin{bmatrix}
                1\\1\\0\\0\\1
            \end{bmatrix},
            A = \begin{bmatrix}
                0 & 1 & 1 & 0 & 0\\
                1 & 0 & 0 & 1 & 0\\
                1 & 0 & 0 & 1 & 1\\
                0 & 1 & 1 & 0 & 0\\
                0 & 0 & 1 & 0 & 0\\
                \end{bmatrix}$};
        \end{tikzpicture}
       \end{minipage}%

    \caption{%
    Example instance of \PrbMaxEHS with $\nd=5$.
    The figure shows the quotient graph $H$ of the given graph $G$ on its modules $M_1,\ldots,M_5$.
    Every edge in $H$ forms a biclique in $G$.
    The maximum edge happy set for $k=10$ are shaded in gray.
    It also shows the number of internal edges for each module (e.g. $(6)$ for $M_1$),
    and that of external edges between modules (e.g. $12$ between $M_1$ and $M_2$).
    Vector $q$ indicates if each module is a clique or an independent set
    (e.g. $q_1=1$ because $M_1$ is a clique),
    and $A$ is the adjacency matrix of the quotient graph.
    }
    \label{fig:alg-ehs-nd}
\end{figure} 

\cref{fig:alg-ehs-nd} exemplifies a quotient graph of a graph with $\nd=5$,
 along with vector $q$ and matrix $A$.
The following is a direct result from \cref{lem:ehs-nd-reduction} and \cref{prop:iqp}.

\begin{theorem}
    \PrbMaxEHS can be solved in time $f(\nd) \cdot \abs{V(G)}^{\mathcal{O}(1)}$,
    where $\nd$ is the neighborhood diversity of the input graph $G$
    and $f$ is a computable function.
\end{theorem}

\subsection{Parameterized by cluster deletion number}

Finally, we present an FPT algorithm for \PrbMaxEHS parameterized
 by the cluster deletion number of the given graph.

\begin{theorem}\label{thm:ehs-cd}
    Given a graph $G=(V,E)$ and its cluster deletion set $X$ of size $\cd$,
    \PrbMaxEHS can be solved in time $\bigo{2^{\cd}\cdot k^2 \cdot \abs{V}}$.
\end{theorem}

Recall that by definition, $G[V \setminus X]$ is a set of disjoint cliques.
Let $C_1, \ldots, C_p$ be the clusters appeared in $G[V \setminus X]$.
Our algorithm first guesses part of the happy set $S'$, defined as $S \cap X$,
 and performs \PrbFKnapsack with $p$ items.

\begin{alg}[\AlgMaxEHSCD]
    \label{alg:max-ehs-cd}
    Given a graph $G=(V,E)$ and its cluster deletion set $X$,
     consider all sets of $S' \subseteq X$ such that $\abs{S'} \leq k$
     and proceed the following steps.

    (Step 1) For each clique $C_i$, sort its vertices in non-increasing order of the number of neighbors in $S'$.
    Let $v_{i,1},\ldots,v_{i,\abs{C_i}}$ be the ordered vertices in $C_i$.
    (Step 2) For each $1 \leq i \leq p$, construct a function $f_i$ as follows:
    $f_i(0) = 0$ and for every $1 \leq j \leq \abs{C_i}$,
     $f_i(j) = f_i(j-1) + \abs{N(v_{i,j}) \cap S'} + j - 1$.
    (Step 3) Formulate an \PrbFKnapsack instance with capacity $k - \abs{S'}$
    and value functions $f_i$ for every $1 \leq i \leq p$.
    Then, obtain the solution $\{x_i\}$ with the exact capacity $k - \abs{S'}$ if feasible.
    (Step 4) Construct $S$ as follows:
    Initialize with $S'$ and for each clique $C_i$, pick $x_i$ vertices in order and include them in $S$.
    That is, update $S \gets S \cup \{v_{i,1},\ldots,v_{i,x_i}\}$ for every $1\leq i \leq p$.
    Finally, return $S$ that maximizes $\abs{E(G[S])}$.
\end{alg}

Intuitively, we construct function $f_i$ in a greedy manner.
When we add a vertex $v$ in clique $C_i$ to the happy set $S$,
 it will increase the number of happy edges by
 the number of $v$'s neighbors in $S'$ and
 the number of the vertices in $C_i$ that are already included in $S$.
Therefore, it is always advantageous to pick a vertex having the most neighbors in $S'$.
\cref{fig:alg-ehs-cd} illustrates the key ideas of \cref{alg:max-ehs-cd}.
The following proposition completes the proof of \cref{thm:ehs-cd}.

\begin{figure}[t]
    \centering
    \pgfdeclarelayer{bg}
    \pgfsetlayers{bg, main}
    \usetikzlibrary{calc}

    \tikzset{
        old inner xsep/.estore in=\oldinnerxsep,
        old inner ysep/.estore in=\oldinnerysep,
        double circle/.style 2 args={
            circle,
            old inner xsep=\pgfkeysvalueof{/pgf/inner xsep},
            old inner ysep=\pgfkeysvalueof{/pgf/inner ysep},
            /pgf/inner xsep=\oldinnerxsep+#1,
            /pgf/inner ysep=\oldinnerysep+#1,
            alias=sourcenode,
            append after command={
            let     \p1 = (sourcenode.center),
                    \p2 = (sourcenode.east),
                    \n1 = {\x2-\x1-#1-0.5*\pgflinewidth}
            in
                node [inner sep=0pt, draw, circle, minimum width=2*\n1,at=(\p1),#2] {}
            }
        },
        double circle/.default={-3pt}{black}
    }
    \tikzstyle{small} = [circle, fill=white, text=black, draw, thick, scale=1, minimum size=0.2cm, inner sep=1.5pt]
    \tikzstyle{happy} = [circle, fill=black!25, text=black, draw, thick, scale=1, minimum size=0.2cm, inner sep=1.5pt]
    \tikzstyle{large} = [circle, fill=white, text=black, draw, thick, scale=1, minimum size=1.2cm, inner sep=1.5pt]
    \begin{minipage}{0.96\textwidth}
        \centering
        \footnotesize
        \begin{tikzpicture}
            \node[happy] (x1) at (3, 2) {};
            \node[happy] (x2) at (4, 2) {};
            \node[happy] (x3) at (5, 2) {};
            \node[small] (x4) at (6, 2) {};
            \node[small] (x5) at (7, 2) {};
            \node[small] (x6) at (8, 2) {};
            
            \node[small, label={[shift={(0,-1.0)}]$v_{1,1}$}] (v1) at  (3, 0.3) {};
            \node[small, label={[shift={(0,-1.0)}]$v_{1,2}$}] (v2) at  (4, 0.3) {};
            \node[small, label={[shift={(0,-1.0)}]$v_{1,3}$}] (v3) at  (5, 0.3) {};
            \node[small, label={[shift={(0,-1.0)}]$v_{1,4}$}] (v4) at  (6, 0.3) {};
            \node[small, label={[shift={(0,-0.7)}]$v_{2,1}$}] (u1) at  (7, 0.3) {};
            \node[small, label={[shift={(0,-0.7)}]$v_{2,2}$}] (u2) at  (8, 0.3) {};

            \draw[very thick] (x1) -- (v1);
            \draw[very thick] (x2) -- (v1);
            \draw[very thick] (x3) -- (v1);

            \draw[very thick] (x2) -- (v2);
            \draw[very thick] (x3) -- (v2);
            \draw (x4) -- (v2);
            
            \draw[very thick] (x3) -- (v3);
            
            \draw (x1) -- (x2);
            \draw (x3) -- (x4);
            \draw (x5) -- (x6);

            \draw (x4) -- (v4);
            \draw (x5) -- (v4);
            
            \draw (v1) -- (v2);
            \draw (v2) -- (v3);
            \draw (v3) -- (v4);
            \draw (u1) -- (u2);
            \draw (v1) to[in=205, out=-25] (v3);
            \draw (v2) to[in=205, out=-25] (v4);
            \draw (v1) to[in=210, out=-30] (v4);
            
            \draw[very thick] (x2) -- (u1);
            \draw[very thick] (x3) -- (u1);
            \draw (x6) -- (u2);

            \node at (1.9, 2.6) {$X$};
            \node at (2.8, 2.55) {$S'$};
            \node at (2.3, 0.2) {$C_1$};
            \node at (8.7, 0.2) {$C_2$};
            \draw[rounded corners, dashed] (2.5, 1.7) rectangle ++ (3, 0.6);
            \draw[rounded corners] (2.2, 1.5) rectangle ++ (6.4, 1.4);
            \draw[rounded corners] (2.6, -0.7) rectangle ++ (3.8, 1.3);
            \draw[rounded corners] (6.6, -0.7) rectangle ++ (1.8, 1.3);

            \node at (1.2, 0.5) {
                $\begin{aligned}
                    f_1(0) &= 0\\
                    f_1(1) &= 3\\
                    f_1(2) &= 6\\
                    f_1(3) &= 9\\
                    f_1(4) &= 12
                \end{aligned}$
            };

            \node at (9.8, 0) {
                $\begin{aligned}
                    f_2(0) &= 0\\
                    f_2(1) &= 2\\
                    f_2(2) &= 3
                \end{aligned}$
            };
        \end{tikzpicture}
       \end{minipage}%

    \caption{%
    Visualization of \AlgMaxEHSCD,
    given a graph $G=(V,E)$ with its cluster deletion set $X$ and a fixed partial solution
    $S'$ ($\abs{S'}=3$, shaded in gray).
    The graph after removing $X$, i.e. $G[V \setminus X]$, forms cliques $C_1$ and $C_2$.
    For each clique, vertices are sorted in decreasing order of the number of neighbors in $S'$
    (edges to $S'$ in thicker lines).
    Functions $f_1$ and $f_2$ are constructed as described in the algorithm and used for \PrbFKnapsack.
    For example, $f_1(3) = f_1(2) + 1 + (3-1)$ as vertex $v_{1,3}$ has one edge to $S'$ and
    two edges to previously-added $v_{1,1}$ and $v_{1,2}$.
    If $k=6$, then we pick $k-\abs{S'}=3$ vertices from $C_1$ and $C_2$.
    The optimal solution would be $\{v_{1,1}, v_{1,2}, v_{1,3}\}$
    because $f_1(3) + f_2(0) = 9$ gives the maximum objective value in the \PrbFKnapsack formulation.
    }
    \label{fig:alg-ehs-cd}
\end{figure}
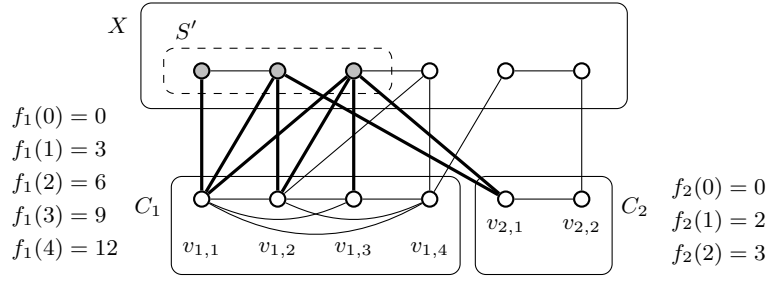

\begin{proposition}
    Given a graph $G=(V,E)$ and its cluster deletion set $X$ of size $\cd$,
    \cref{alg:max-ehs-cd} correctly finds the maximum edge happy set in time $\bigo{2^{\cd} \cdot k^2 \cdot \abs{V}}$.
\end{proposition}

\begin{proof}
    The algorithm considers all possible sets of $S \cap X$, so the optimal solution should extend
    one of them.
    It is clear to see that when the \PrbFKnapsack instance is feasible,
     $S$ ends up with $k$ vertices, since the sum of the obtained solution must be $k-\abs{S'}$.
    The objective of the \PrbFKnapsack is equivalent to $\abs{E(G[S])} - \abs{E(G[S'])}$,
    that is, the number of happy edges extended by the vertices in $V \setminus X$.
    Since $S'$ has been fixed at this point, the optimal solution to \PrbFKnapsack
    leads to that to \PrbMaxEHS.
    Lastly, the value function $f_i$ is correct because for each clique $C_i$,
    the number of extended edges is given by $\binom{\abs{S_i}}{2} + \sum_{v \in S_i} \abs{N(v) \cap S'}$,
    where $S_i = S \cap C_i$ and $x_i = \abs{S_i}$.
    This is maximized by choosing $\abs{S_i}$ vertices that have the most neighbors in $S'$,
    if we fix $\abs{S_i}$, represented as $x_i$ in \PrbFKnapsack.
    This is algebraically consistent with the recursive form in step 2.

    For the running time, the choice of $S'$ adds the complexity of $2^{\cd}$ to the entire algorithm.
    Having chosen $S'$, vertex sorting (step 1) can be accomplished by
    checking the edges between $S'$ and $V \setminus X$, so it takes only $\bigo{k\cdot |V|}$.
    The \PrbFKnapsack (step 3) takes time $\bigo{pk^2} = \bigo{k^2\cdot |V|}$ from \cref{cor:f-knapsack},
    because there are $p$ items and weights are bounded by $k$.
    Steps 2 and 4 do not exceed this asymptotic running time.
    The total runtime is $\bigo{2^{\cd} \cdot k^2 \cdot \abs{V}}$.\looseness-1
\end{proof}

\section{Conclusions \& Future Work}\label{sec:conclusion}

We present four algorithms using a variety of techniques,
 two for \PrbMaxHSLong (\PrbMaxHS) and two for \PrbMaxEHSLong (\PrbMaxEHS).
The first shows that \PrbMaxHS is FPT with respect to
 the modular-width parameter, which is stronger than clique-with but generalizes several parameters
 such as neighborhood diversity and twin cover number.
We then give an FPT dynamic-programming algorithm for \PrbMaxHS parameterized by clique-width.
This improves the best known complexity result of FPT when parameterized by
 clique-width plus $k$.

For \PrbMaxEHS, we prove that it is FPT by neighborhood diversity, using \PrbIQPLong.
Lastly, we show an FPT algorithm parameterized by cluster deletion number,
the distance to a cluster graph, which then implies FPT by twin cover number.
These results have resolved several open questions of \cite{asahiro2021parameterized} (\cref{tab:result}).

There are multiple potential directions for future research.
As highlighted in \cref{tab:result}, the parameterized complexity of \PrbMaxEHS
 with respect to modular-width is still open.
Another direction would be to find the lower bounds for known algorithms.

\bibliography{main}


\end{document}